\documentclass[11pt]{amsart}
\usepackage[left=1.8cm,right=1.8cm]{geometry}
\usepackage{amsmath,amsthm,amsfonts,amssymb}
\usepackage{verbatim}
\usepackage{latexsym}
\usepackage{nicefrac}
\usepackage{color}
\usepackage{enumerate}
\usepackage{mdwlist}
\usepackage[british]{babel}
\usepackage[utf8]{inputenc}
\usepackage{soul}
\usepackage{hyperref}
\newcommand{\R}{{\mathbb{R}}}
\newcommand{\C}{{\mathbb{C}}}

\newcommand{\N}{\mathbb{ N}}
\newcommand{\Z}{\mathbb{ Z}}
\newcommand{\T}{\mathbb{ T}}
\newcommand{\leb}{\lambda\raisebox{0,06cm}{\!\!\!\text{\scriptsize$\setminus$}}\,}

\newcommand{\intpar}[4]{\int\limits_{#1}^{#2}#3~\mathrm{d}#4}

\newtheorem{theorem}{Theorem}[subsection]
\newtheorem{lemma}[theorem]{Lemma}
\newtheorem{proposition}[theorem]{Proposition}
\newtheorem{definition}[theorem]{Definition}
\newtheorem{corollary}[theorem]{Corollary}

\newtheorem{example}[theorem]{Example}

\newtheorem{remark}[theorem]{Remark}

\begin{document}
	
	\title{Sparse Recovery from Group Orbits}
	\author{Timm Gilles, Hartmut Führ}
	\email{gilles@mathga.rwth-aachen.de, fuehr@mathga.rwth-aachen.de}
	\address{Lehrstuhl f\"ur Geometrie und Analysis, RWTH Aachen University, D-52056 Aachen, Germany}
	
	\maketitle
	\begin{abstract}
        While most existing sparse recovery results allow only minimal structure within the measurement scheme, many  practical problems possess significant structure. To address this gap, we present a framework for structured measurements that are generated by random orbits of a group representation associated with a finite group. We differentiate between two scenarios: one in which the sampling set is fixed and another in which the sampling set is randomized. For each case, we derive an estimate for the number of measurements required to ensure that the restricted isometry property holds with high probability. These estimates are contingent upon the specific representation employed. For this reason, we analyze and characterize various representations that yield favorable recovery outcomes, including the left regular representation. Our work not only establishes a comprehensive framework for sparse recovery of group-structured measurements but also generalizes established measurement schemes, such as those derived from partial random circulant matrices.
	\end{abstract}

	\vspace{0.4cm}
	
	
	\noindent{\small {\bf MSC-class:} {\em Primary:} 15B52, 20C35, 94A12. {\em Secondary:} 60G50, 60G70.}
	
\section{Introduction}

\subsection{Sparse Recovery}
	
	Sparse Recovery \cite{candes2006robust,donoho2006compressed}, also known as Compressed Sensing, is a technique in signal processing to recover a sparse signal via an efficient algorithm from fewer measurements than traditional methods would use. To be mathematically precise: Let $\Phi=(\phi_1,\ldots,\phi_m)^{\ast}\in\C^{m\times n}$ be some matrix, also called measurement matrix, and let $y_1,\ldots,y_m\in\C$ be given measurements that arise from a known linear measurement process \[y_j=\langle x,\phi_j\rangle\]
	of some signal $x\in\C^n$. This can also be written as $y=\Phi x$. The goal of sparse recovery is to reconstruct the signal $x$ from $m\ll n$ measurements. Clearly, this is not possible in general and consequently one needs another condition to make the recovery problem well-posed. The condition required is the a priori knowledge that $x$ is $s$-sparse, i.e. $x$ has at most $s$ non-zero entries.\par
	One could try solving the sparse recovery problem via the minimization problem
	\begin{align}\label{eq:l1_min_problem}
	\min \Vert z\Vert_0\quad\text{subject to} \,\,\Phi z=y 
	\end{align}
	where $\Vert x\Vert_0:=\vert\{ j\in\{1,\ldots, n\}\mid x_j\neq 0\}\vert$. Unfortunately, this is NP-hard in general \cite{natarajan1995sparse}. Nevertheless, it has been shown that if one imposes some additional conditions on the measurement matrix $\Phi$, then $\ell^1$-minimization
	\[
	\min \Vert z\Vert_1\quad\text{subject to} \,\,\Phi z=y 
	\]
	reconstructs every $s$-sparse vector $x$ uniquely from $\Phi x=y$ \cite{foucart2013mathematical}. If this reconstruction holds, we say that $\Phi$ does $s$-sparse recovery. There are also many other algorithms that enable recovery under these conditions (e.g. OMP, CoSaMP or IHT \cite{tropp2007signal, needell2009cosamp,blumensath2009iterative}).
	Various such conditions on the measurement matrix, such as the null space property, a small coherence of $\Phi$ or the restriced isometry property (RIP), have been studied \cite{tropp2004greed, candes2005decoding}. All of these can guarantee recovery via $\ell^1$-minimization \cite{foucart2013mathematical}. In this work, we will primarily use the RIP. For a matrix $\Phi\in\C^{m \times n}$ and $1\leq s\leq n$, the restricted isometry constant $\delta_s$ is the smallest $\delta\geq 0$ such that
	\[
	(1-\delta)\Vert x\Vert_2^2\leq \Vert \Phi x\Vert_2^2 \leq(1+\delta) \Vert x\Vert_2^2 
	\]
for all $s$-sparse vectors $x\in\C^n$. If $\delta_{2s}<0.4931$, then $\ell^1$-minimization reconstructs every $s$-sparse vector $x$ uniquely from $\Phi x=y$ \cite{mo2011new}.\par
	One of the main questions arising in the field of compressed sensing is determining which types of measurement matrices are suited best for recovery, e.g. by fulfilling the restricted isometry property. Most known results use fully random measurement matrices $\Phi$, since these are the only ones for which optimal measurement bounds have been proven. By ``optimal bound'' we refer to the minimum number of measurements necessary to enable stable $s$-sparse reconstruction through the minimization problem (\ref{eq:l1_min_problem}). This optimal bound \cite{foucart2010gelfand,foucart2013mathematical} is given by
	\begin{align}\label{eq:optimal_bound_on_m}
	m\gtrsim s \ln\left(\frac{en}{s}\right).
	\end{align}
	In practice one generally does not have completely randomized measurements and consequently one is interested in reducing the degree of randomness in the measurements and allowing some kind of structure.\par 
	In this work, we address this problem by establishing a comprehensive framework for sparse recovery of group-structured measurement schemes where each measurement arises as a random orbit. This drastically reduces the degree of randomness in the measurement scheme and can hence be understood as a partial derandomization of the measurement process.\par 
	As a special case this setup contains partial random circulant matrices which arise in different applications such as system identification and active imaging \cite{haupt2010toeplitz,romberg2009compressive} and which have been studied in \cite{romberg2009compressive, rauhut2012restricted, krahmer2014suprema}.
	
\subsection{Recovery from group orbits: The setup}

	We allow the vector $x$, we want to recover, to be sparse in some known orthonormal basis $B=(b_1\mid\ldots\mid b_{n}) \in\C^{n \times n}$, i.e. \[x=B z\] and $z\in\C^n$ is $s$-sparse. This generalizes the problem described above; choosing $B=I_n$ yields the classical sparse recovery setup.\par
	We will study measurements that are generated by the inner product of the unknown vector $x$ and some element of a group orbit. More precisely, let $G$ be some finite group and $\pi \colon G\to\text{GL}(\C^n)$ a unitary (projective) representation (see Section \ref{subsec:intro_group_rep} for more information on group representations). Further, let $\xi\in \C^n$ be some vector, which we will call generating vector and which usually will be randomized in our results.	We assume that we can measure the signal $x$ via the inner product
	$\langle x,\pi(g)\xi\rangle$ for any $g\in G$. However, only the measurements at some sampling points $\omega_1,\ldots,\omega_m$ are known. Hence, our sampled measurements are given by
	\begin{equation*}
		y_g=\langle x,\pi(g)\xi\rangle
	\end{equation*}  
	for $g\in\Omega=\{\omega_1,\ldots,\omega_m\}$. We differentiate between two approaches of choosing the sampling sets:
	\begin{enumerate}[(i)]
		\item\label{item:fixed_sampling_sets} Depending on the representation $\pi$ one fixes a set of possible sampling sets $\widetilde{\Omega}\subseteq \mathcal{P}(G)$. Then, one chooses a sampling set 
		$\Omega=\{\omega_1,\ldots,\omega_m\}\in\widetilde{\Omega}$.
		\item\label{item_random_omega} The sampling points $\omega_1,\ldots,\omega_m$ are selected independently and uniformly at random from $G$.
	\end{enumerate}
	The reasoning of this distinction and its advantages will become more apparent in the next section when we discuss the main results of this work. Independent of how $\Omega=\{\omega_1,\ldots,\omega_m$\} is chosen, the measurement matrix of the described problem can be written as
	\begin{align}\label{eq:def_phi_matrix}
		\Phi_{\pi}=\frac{1}{\sqrt{m}}\, R_{\Omega}\big(\pi(g) \xi\big) ^{\ast}_{g\in G}B\in \C^{m\times n}
	\end{align}
	where $R_{\Omega}\colon \C^G\mapsto \C^m$ restricts a vector to its entries in $\Omega$, i.e. $(R_{\Omega}y)_l=y(\omega_l)$. Here, the factor $\frac{1}{\sqrt{m}}$ is needed for normalization reasons. When the choice of $\pi$ is apparent from the context, we will just write $\Phi$. In this work, we examine which groups and representations can be utilized to construct measurement matrices of the type of (\ref{eq:def_phi_matrix}) and that do $s$-sparse recovery. By that we mean that we are interested in results of the following form: For $\pi$, let $\Phi_{\pi}$ be the measurement matrix associated with $\pi$ as defined in (\ref{eq:def_phi_matrix}). If $m\gtrsim f_{\pi}(s,n)$ holds for a suitable function $f_{\pi}$, then $\Phi_{\pi}$ does $s$-sparse recovery with high probability. The function $f_{\pi}$ quantifies the impact of the selected representation on the number of measurements required. It is particularly interesting to compare our derived bounds with the optimal bound (\ref{eq:optimal_bound_on_m}) to understand the trade-off involved in reducing randomness within the measurement scheme.\par
	To compare the recovery properties of different representations, we say that two representations $\pi$ and $\rho$ have the same $s$-sparse recovery properties if there exist $\xi_{\pi}$ and $\xi_{\rho}$ as well as a unitary basis $B$ such that the associated matrices $\Phi_{\pi}$ and $\Phi_{\rho}$ do $s$-sparse recovery (with high probability). A natural follow-up question is whether unitarily equivalent representations (i.e. different realizations of one representation) have the same $s$-sparse recovery property. Note that for $\pi$ and $\rho$ unitarily equivalent there exists a unitary matrix $V\in\C^{n\times n}$ such that $\pi(g)=V^{\ast} \rho(g)V$ for all $g\in G$. Fix some unitary matrix $B\in\C^{n\times n}$. Then,
	\begin{align}\label{eq:relation_meas_mat_equiv_rep}
		\langle Bx,\pi(g)\xi\rangle=\langle VBx,\rho(g)V\xi\rangle.
	\end{align}
	Hence, \[\frac{1}{\sqrt{m}} R_{\Omega} \big(\rho(g)V\xi\big)_{g\in G}^{\ast} VB=\frac{1}{\sqrt{m}} R_{\Omega} \big(\pi(g)\xi\big)_{g\in G}^{\ast} B\]
	and both matrices are of the type (\ref{eq:def_phi_matrix}). But even if one of the matrices does $s$-sparse recovery, this does not imply that $\pi$ and $\rho$ have the same $s$-sparse recovery properties because the considered bases are different. In this work, we aim to give an answer to the question under which assumptions unitarily equivalent representations have the same $s$-sparse recovery properties.

\subsection{Overview of the main results}\label{sec:overview_results}

	Our first main result establishes the restricted isometry property for a random subgaussian generating vector $\xi$ and fixed sampling sets as in (\ref{item:fixed_sampling_sets}). Ideally, one would hope that any sufficiently large sampling set would suffice to establish the RIP, implying that the only relevant characteristic of a sampling set is its size. This is, for example, the case for measurements from partial random circulant matrices \cite{krahmer2014suprema}. However, we have observed that different sampling sets yield varying performance across certain representations. This motivated the choice of sampling points in (\ref{item:fixed_sampling_sets}): We restrict our sampling set $\Omega=\{\omega_1,\ldots,\omega_m\}$ to be an element of a specific subset of $\mathcal{P}(G)$, denoted as $\widetilde{\Omega}$. This set $\widetilde{\Omega}$ can be thought of as the set of possible sampling sets. Then, our result states that every sufficiently large sampling set $\Omega\in\widetilde{\Omega}$ establishes the RIP. Hence, the only relevant characteristic of a sampling set within $\widetilde{\Omega}$ is its size. This more refined approach of formulating a recovery result, in contrast to most literature, allows us to derive bounds on the number of measurements for representations where sparse recovery would otherwise not be feasible. We will elaborate on this point further after presenting our first result. Note that we further assume that $x$ is sparse in the standard basis, i.e., $B=I_n$.
	\begin{theorem}\label{mainresult1}
			Let $\pi$ be a unitary (projective) representation. Let $\xi$ be a random vector with independent, mean 0, variance 1, and $L$-subgaussian entries. Fix $\widetilde{\Omega}\subseteq \mathcal{P}(G)$ and let  $C_{\widetilde{\Omega},\pi}>0$ be a constant such that
			\begin{align}\label{main_annahme}
				\sup_{1\leq j\leq n} \Big\Vert  R_{\Omega}\big(\pi(g) y\big) ^{\ast}_{g\in G} e_j 	\Big\Vert_{2} \leq \sqrt{C_{\widetilde{\Omega},\pi}} \Vert y\Vert_2
			\end{align}
			holds for all $y\in\C^n$ and $\Omega\in \widetilde{\Omega}$. Then, for all sampling sets $\Omega\in\widetilde{\Omega}$ it holds: If, for $s\leq n$ and $\delta,\eta \in(0,1)$,
			\begin{align}\label{gl:absm}
				m=\vert \Omega\vert \geq c\delta^{-2} s \,C_{\widetilde{\Omega},\pi} \, \max\Big\{(\ln(s\, C_{\widetilde{\Omega},\pi}))^2 ((\ln n)(\ln 4n),\,\ln(\eta^{-1})\Big\},
			\end{align}
			then with probability at least $1-2\eta$, the restricted isometry constant of
			\begin{equation*}
				\Phi=\frac{1}{\sqrt{m}}\, R_{\Omega}\big(\pi(g) \xi\big) ^{\ast}_{g\in G}\in \C^{m\times n}
			\end{equation*}
			satisfies $\delta_s\leq \delta$. Here, $c>0$ only depends on $L$.
		\end{theorem}
		 
		Let us further discuss the set $\widetilde{\Omega}$ and the constant $C_{\widetilde{\Omega},\pi}$ which appear in Theorem \ref{mainresult1}. The constant depends on the set $\widetilde{\Omega}\subseteq\mathcal{P}(G)$ of possible sampling sets as well as the representation $\pi$. Its presence in the theorem arises from the need to bound the term
		\begin{align*}
			\sup_{1\leq j\leq n}\Big\Vert R_{\Omega}\big(\pi(g) y\big) ^{\ast}_{g\in G} e_j 	\Big\Vert_{2}
		\end{align*}
		in terms of $\Vert y\Vert_2$ within the proof of Theorem \ref{mainresult1}. The existence of a constant $C_{\widetilde{\Omega},\pi}>0$ such that
		\begin{align*}
			\sup_{1\leq j\leq n} \Big\Vert R_{\Omega}\big(\pi(g) y\big) ^{\ast}_{g\in G} e_j 	\Big\Vert_{2} \leq \sqrt{C_{\widetilde{\Omega},\pi}} \Vert y\Vert_2
		\end{align*}
		holds for all $\Omega\in\widetilde{\Omega}$ is obvious, since we are considering only finite dimensional spaces. However, what is particularly interesting is the size of $C_{\widetilde{\Omega},\pi}$. This constant essentially encapsulates how effectively the measurement matrix $\Phi$ of a given representation performs sparse recovery across any of the sampling sets $\Omega\in\widetilde{\Omega}$, since the number of needed measurements depends linearly on it.\par
		Ideally one would wish for representations such that the associated measurement matrices do $s$-sparse recovery for any sampling set $\Omega\subseteq G$ of sufficient size. Translating this wish into the context of the constant $C_{\widetilde{\Omega},\pi}$, it would imply two key points. Firstly, $\widetilde{\Omega}=\mathcal{P}(G)$, as this ensures that any sampling set that is sufficiently large (quantified by (\ref{gl:absm})) performs equally well or poorly. Secondly,  $C_{\widetilde{\Omega},\pi}$ should be close or equal to $1$ indicating that the bound in (\ref{gl:absm}) is near optimal. Comparing the bound in (\ref{gl:absm}) for $C_{\widetilde{\Omega},\pi}=1$ with the optimal bound in (\ref{eq:optimal_bound_on_m}) shows that we are off by some logarithmic factors. This appears to be the price one must pay when imposing structure on the measurement scheme and is consistent with established bounds for structured measurements in the literature \cite{krahmer2014suprema,rudelson2008sparse}. That $1$ is the smallest possible value for $C_{\widetilde{\Omega},\pi}$ is easy to see: For any non-empty sampling set $\Omega\in\widetilde{\Omega}$ and $1\leq j\leq n$ we can choose $g_0\in\Omega$ and $y=\pi(g_0)^{-1}e_j$. Then,
		\[
		\Big\Vert R_{\Omega}\big(\pi(g) y\big) ^{\ast}_{g\in G} e_j \Big\Vert_{2}=\sum_{g\in\Omega}\vert \langle e_j,\pi(g)y\rangle\vert^2=1+\sum_{g\in\Omega\setminus\{g_0\}}\vert\langle e_j,\pi(g)y\rangle\vert^2\geq 1.
		\]
		A natural follow-up question is: \textit{1) For what representations does $C_{\mathcal{P}(G),\pi}=1$ hold?} For most representations, the requirement $C_{\mathcal{P}(G),\pi}=1$ is too strict. For example, consider the trivial representation $\pi=I_n$. In this case, the constant $C_{\mathcal{P}(G),I_n}$ is given by $\vert G\vert$ since it holds
		\begin{align*}
			\Big\Vert R_{\Omega}\big(I_n y\big) ^{\ast}_{g\in G} e_j 	\Big\Vert_{2}^2=\sum_{g\in\Omega} \vert \langle e_j,y\rangle\vert^2=\vert \Omega\vert \vert y_j\vert^2
		\end{align*}
		for all $y\in\C^n$, all canonical vectors $e_j\in\C^n$ and all
 		$\Omega\in\mathcal{P}(G)$. Hence, the smallest possible choice for $C_{\mathcal{P}(G),I_n}$ is $\vert G\vert$. This makes the inequality on the number of measurements given in (\ref{gl:absm}) unsatisfiable. In Section \ref{subsect:what_to_expect}, we will show that this is no coincidence by proving that $s$-sparse recovery is simply not possible for the measurement matrix that is associated with the trivial representation. However, one can ask a weaker version of the first question: \textit{2) What representations admit a constant $C_{\widetilde{\Omega},\pi}$ that is close or equal to 1 for a large set of possible sampling sets $\widetilde{\Omega}\subseteq \mathcal{P}(G)$?} The set $\widetilde{\Omega}$ should be chosen to be as large as possible while still ensuring that $C_{\widetilde{\Omega},\pi}$ remains small. It is also important for $\widetilde{\Omega}$ to include sampling sets of various sizes so that Theorem \ref{mainresult1} is applicable, as it requires $\Omega\in\widetilde{\Omega}$ to be of sufficient size. These are the most important aspects to consider when choosing $\widetilde{\Omega}$. These two group-theoretic questions will be primarily discussed in Section \ref{sect:fix_zuf}. Only through this discussion can Theorem \ref{mainresult1} be fully understood.\par
		In Section \ref{subsec:constant_analysis_ind_of_omega}, we will show that there exist representations, such as the left regular representation, for which $C_{\mathcal{P}(G),\pi}=1$ holds. We will also observe that measurement matrices associated with representations that are, in a sense, more closely aligned with the left regular representation exhibit better recovery properties. These first two results will hold for $\widetilde{\Omega}=\mathcal{P}(G)$, meaning that in these cases all sufficiently large sampling sets, quantified by (\ref{gl:absm}), perform equally well. These results also give an answer to our first question.\par
		The second question will be primarily discussed in Section \ref{subsec:constant_analysis_specific_omega}. There, we will consider the case where we no longer require that $\widetilde{\Omega}=\mathcal{P}(G)$. This means that we a priori restrict our sampling sets $\Omega$ to be elements of a fixed subset $\widetilde{\Omega}\subseteq \mathcal{P}(G)$. This is motivated by the observation that for some representations there exist specific choices of sampling sets $\Omega$ which are highly redundant. As a result, $C_{\widetilde{\Omega},\pi}$ is of order $n$, which in turn gives a bad bound on the number of measurements in (\ref{gl:absm}). However, by excluding the highly redundant sampling sets, one can hope for improved bounds. We show that for a specific class of representations, one can a priori determine a subset $\widetilde{\Omega}\subseteq \mathcal{P}(G)$ such that $C_{\widetilde{\Omega},\pi}=1$ holds for every $\Omega\in \widetilde{\Omega}$.\par
		Furthermore, we show in Section \ref{subsec:constant_analysis_ind_of_omega} that the size of $C_{\widetilde{\Omega},\pi}$ depends on the realization of the representation. This implies that in the setting of Theorem \ref{mainresult1} unitarily equivalent representations do not have the same $s$-sparse recovery properties. Conversely, representation theory examines properties of representations that remain invariant under unitary basis transformations. This motivates the question of whether the measurement process can be adapted such that equivalent representations have the same $s$-sparse recovery properties. Our second main result answers this question. The key modification is that the sampling set $\Omega$ is randomized, i.e. we are in the setting of (\ref{item_random_omega}). This allows us to establish the RIP for a measurement scheme in which the vector $x$ can be $s$-sparse in any basis $B$.

	\begin{theorem}\label{mainresult2_low_key}
	Let $\pi$ be a unitary representation given in block-diagonal form with $m_{\pi}(\rho)=1$ for all irreducible subrepresentations $\rho\leq \pi$ (see Section \ref{subsec:intro_group_rep}). Let $\xi\in\C^n$ be a random vector with independent entries that each are uniformly distributed on the torus $\mathbb{T}$. Let the measurement matrix be
	\[
	\Phi=\frac{1}{\sqrt{m}} R_{\Omega}\left(\pi(g)\xi\right)_{g\in G}^{\ast} B\in\C^{m \times n}
	\]
	where $B\in\C^{n\times n}$ is unitary and $\Omega=(\omega_1,\ldots,\omega_m)$ is a sequence of independent random variables with $\omega_i \sim \mathcal{U}(G)$ such that $\Omega$ and $\xi$ are independent.
	If, for $\eta,\delta \in(0,1)$,
		\begin{align*}
			m\geq C \delta^{-2} s \ln(8\vert G\vert)\ln\left(\frac{2}{\eta}\right) \max\bigg\{ &\ln(4s)^2\ln(8n)
		\ln\left(s\delta^{-2}\ln(8\vert G\vert)\ln\left(\frac{2}{\eta}\right) \right),\ln\left(\frac{2}{\eta}\right)\bigg\},
		\end{align*}
		then with probability at least $1-\eta $ the restricted isometry constant $\delta_s$ of $\Phi$ satisfies \[\delta_s\leq 3\delta.
		\]
		Here, $C>0$ is an absolute constant.
	\end{theorem}
	
	In Section \ref{sect:random_xi_and_omega}, we will prove a more general version of the above theorem which also allows for multiplicities within the representation $\pi$. \par
	By randomizing the sampling set and allowing $x$ to be sparse in some basis $B$, Theorem \ref{mainresult2_low_key} implies that the measurement matrices of equivalent representations have the same $s$-sparse recovery properties. This is not the case when $\Omega$ is a fixed sampling set. We provide detailed arguments for this in Example \ref{bsp:F_faltung_F} and Remark \ref{rem:equiv_rep}.\par	
	Another consequence of this finding is that there exist representations for which Theorem \ref{mainresult2_low_key} is applicable with good bounds on the number of measurements, even though sparse recovery was generally not possible for the measurement matrices that are associated with these representations when only the generating vector was randomized. As an example, consider the group $G=\Z/n\Z$. The representation $\rho\colon G\to\text{GL}(\C^n)$ defined by
	\begin{align*}
	\rho(k)=\begin{pmatrix}
					 e^{\frac{2\pi ik1}{n}}\\
					& \ddots\\
					&&e^{\frac{2\pi ikn}{n}}\\
				\end{pmatrix}
	\end{align*}
	for all $k\in\Z/n\Z$, satisfies the requirements of Theorem \ref{mainresult2_low_key}. Consequently, with probability at least $1-n^{-1}$ roughly 
	\begin{align*}
	m\gtrsim s \ln(s)^2\ln(n)^3 \ln(s\ln(n))
	\end{align*}
	measurements are sufficient for recovery via $\ell_1$-minimization. When $\Omega$ is not randomized, $\Phi_{\rho}$ does not do $s$-sparse recovery, as discussed in Example \ref{bsp:F_faltung_F}. In general, however, none of the aforementioned theorems is stronger than the other; which of the two theorems provides better bounds on the number of measurements depends on the chosen representation. A notable advantage, however, of Theorem \ref{mainresult2_low_key} over Theorem \ref{mainresult1} is that the basis in which the vector $x$ is sparse, i.e the matrix $B$, can be chosen arbitrarily without affecting the recovery outcome.

\subsection{Notation}\label{subsec:notation}
	Let $C,c,C_1,c_1,C_2,c_2,\ldots>0$ denote absolute constants whose values may vary from line to line. $A\lesssim B$ means that there exists a universal constant $c>0$ such that $A\leq c B$. Let $s,k,m,n\in\N$. The torus, i.e. the complex numbers of absolute value 1, is denoted by $\mathbb{T}$. For $S\subseteq \{1,\ldots,n\}$ and $x\in\C^n$ we denote by $x_S$ the vector in $\C^n$ with 
			\begin{equation*}
				(x_S)_j:=\begin{cases}
					x_j,&j\in S,\\
					0,&j\notin S.
				\end{cases}
			\end{equation*}
			We define the conjugate of a matrix $A\in\C^{m\times n}$ by $\left(\overline{A}\right)_{kl}:=\overline{A}_{kl}$. The Fourier transform on $\C^n$ is given by \[\mathcal{F}\colon\C^n\to\C^n,\, (\mathcal{F}x)_l:=\sum_{j=1}^n x_j\, e^{\frac{2\pi i j l }{n}}\] with inverse \[\mathcal{F}^{-1}\colon\C^n\to\C^n,\, (\mathcal{F}^{-1}x)_j:=\frac{1}{n}\sum_{l=1}^n x_l\, e^{-\frac{2\pi i j l }{n}}.\]
			A random vector $\xi\in\C^n$ is called $L$-subgaussian for $L>0$ if for every $x\in\C^n$ with $\Vert X\Vert_2=1$, it holds
			\[
			\mathbb{E}\vert\langle\xi,x\rangle\vert^2=1\,\,\quad\text{and}\quad\,\,\mathbb{P}(\vert \langle \xi,x\rangle\vert\geq t)\leq 2\exp\left(-\frac{t^2}{2L^2}\right)\quad \text{for any } t>0.
			\]
			$G$ denotes a finite group. We will often identify the set of functions $f\colon G\to \C$ with the space $\C^{\vert G \vert}$. Hereby, we implicitly assume some fixed numbering of the group elements $G=(g_1,\ldots,g_{\vert G\vert})$ such that for $x\in\C^{\vert G\vert}$ we can write $x(g):=x_g:=x_j$ with $g=g_j$. For a vector $x\in\C^G$ let $\widetilde{x}(g):=\overline{x(g^{-1})}$. The convolution of two vectors $x,y\in\C^G$ is defined as
			\begin{equation*}
				(x\ast y)(g)=\sum_{h\in G} x(h)\,y(g^{-1}h)
			\end{equation*}
			for all $g\in G$.
\subsection{A brief introduction to group representation theory}\label{subsec:intro_group_rep}
		Let $G$ be a finite group and $V$ an $n$-dimensional $\C$-vector space. We understand a representation $\pi$ to be a group homomorphism between $G$ and $\text{GL}(V)$, the set of linear bijective maps $V\to V$. The dimension $n$ of the space $V$ is called degree of $\pi$ and denoted by $d_{\pi}$.\par
		Let $\pi\colon G\to \text{GL}(V)$ be a representation of $G$. A representation $\rho\colon G\to \text{GL}(W)$ is called a subrepresentation of $\pi$, denoted by $\rho\leq \pi$, if $W$ is a subspace of $V$, $\rho(g)w\in W$ for all $w\in W, g\in G$ and $\pi(g)\vert_{W}=\rho(g)$ for all $g\in G$. The representation $\pi$ is called irreducible if its only subrepresentations are $0$ and $\pi$. 
		Two representations $\pi\colon G\to \text{GL}(V_{\pi})$ and $\rho\colon G\to \text{GL}(V_{\rho})$ are equivalent if there exists a linear bijective map $T\colon V_{\pi}\to V_{\rho}$ such that $\pi(g)=T^{-1}\circ \rho(g) \circ T$ holds for all $g\in G$.\par
		Every $\C$-vector space can be equipped with an inner product. We say that $\pi$ is unitary in such inner product space if $\pi(g)$ is a unitary for every $g\in G$. The complete reducibility theorem \cite{terras1999fourier} states that a unitary $\pi$ can be decomposed into irreducible representations, i.e. $\pi$ is equivalent to a representation $m_{\pi}(\pi_1)\pi_1\oplus \ldots\oplus m_{\pi}(\pi_T)\pi_T$ which is defined by
		\begin{align}\label{eq:block_diag_form_pi}
			(m_{\pi}(\pi_1)\pi_1\oplus \ldots\oplus m_{\pi}(\pi_T)\pi_T)(g)=\begin{pmatrix}
				\pi_1(g) \\
				& \ddots\\
				&&\pi_1(g)\\
				&&& \ddots \\
				&&&& \pi_T(g)\\
				&&&&& \ddots \\
				&&&&&& \pi_T(g)\\
			\end{pmatrix}
		\end{align}
		for all $g\in G$ where all $\pi_{\tau}\colon G\to \text{GL}(d_{\pi_{\tau}},\C)$ are unitary, irreducible, not equivalent to each other and occur $m_{\pi}(\pi_j)$ times for all $1\leq \tau\leq T$. The quantity $m_{\pi}(\pi_j)$ is called multiplicity of $\pi_j$ in $\pi$. We say that $\pi$ is given in block-diagonal form if $\pi=m_{\pi}(\pi_1)\pi_1\oplus \ldots\oplus m_{\pi}(\pi_T)\pi_T$.\par 
		Let $\hat{G}$ denote a complete set of inequivalent irreducible unitary representations of $G$. It is known that $\vert \hat{G}\vert $ equals the number of conjugacy classes in $G$ \cite{terras1999fourier}. Define the Fourier transform on $G$ by
		\begin{equation*}
			\mathcal{F}\colon \C^G\to \bigoplus_{\pi\in\hat{G}}\text{End}(V_{\pi}),\, \mathcal{F}f(\pi):=\hat{f}(\pi):=\sum_{g\in G} f(g)\pi(g).
		\end{equation*}
		Plancherel's theorem states
		\begin{equation*}
			\langle f,h\rangle = \frac{1}{\vert G\vert}\sum_{\pi\in\hat{G}} d_{\pi}\,\text{tr}\left(\mathcal{F}f(\pi)\,(\mathcal{F}h)(\pi)^{\ast}\right)
		\end{equation*}
		for all $f,h\in\C^G$ where $\text{tr}$ denotes the trace of an endomorphism. Furthermore, the following inversion formula \cite{fulton2013representation} holds
		\begin{align*}
			f(g)=\frac{1}{\vert G\vert} \sum_{\pi\in\hat{G}} d_{\pi}\,\text{tr}\left((\mathcal{F}f)(\pi)\pi(g^{-1})\right)
		\end{align*}
		for all $f\in \C^G$. Since $\dim\left(\C^G\right)=\dim\left(\bigoplus_{\pi\in\hat{G}} \text{End}(V_{\pi})\right)$ holds, the Fourier transform is linear and Plancherel's theorem holds, the Fourier transform has an inverse which is given by
		\begin{align*}
			\mathcal{F}^{-1}\colon \bigoplus_{\pi\in\hat{G}} \text{End}(V_{\pi}) \to\C^G, \, 	A\mapsto \left(\frac{1}{\vert G\vert} \sum_{\pi\in\hat{G}} d_{\pi}\, \text{tr}\left(A(\pi)\,\pi\left(g^{-1}\right)\right)\right)_{g\in G}. 
		\end{align*}
		For a more comprehensive and detailed introduction to group representations, we refer the reader to \cite{terras1999fourier} and \cite{fulton2013representation}.\par
		More generally than before, we define a projective representation $\sigma$ to be a mapping from $G$ to $\text{GL}(V)$ such that for all $g,h\in G$ there exists a constant $\lambda(g,h)\in\T$ with
		\[
		\sigma(g)\sigma(h)=\lambda(g,h)\sigma(gh).
		\]
		Analogously to before, $\sigma$ is called a unitary projective representation if $\sigma(g)$ is unitary for every $g\in G$.\par
		In this paper, the $\C$-vector space $V$ will usually be a subspace of $\C^k$. Thus, we will frequently identify $\pi(g)$ with its transformation matrix in $\C^{d_{\pi}\times d_{\pi}}$ with respect to the standard basis in $\C^{d_{\pi}}$.

\subsection{Not all representations yield good measurement matrices}\label{subsect:what_to_expect}

Before we start our main analysis in sections \ref{sect:fix_zuf} and \ref{sect:random_xi_and_omega}, the following example might be helpful to understand what characteristics we are looking for in a representation $\pi$ in order to ensure that $\Phi_{\pi}$ does $s$-sparse recovery. 
\begin{example}
	We consider any finite group $G$ and its trivial representation $\pi(g)=I_n$ for all $g\in G$. Let $n\geq 3$. Fix $\xi\in\C^n$ and a unitary matrix $B\in \C^{n\times n}$. Further fix some finite group $G$, some subset $\Omega\subseteq G$ of size $m$ and define $\pi(g)=I_n$ for all $g\in G$.	Then, there exist $1$-sparse vectors $x_1,x_2\in\C^n$ with $x_1\neq x_2$ such that
	\begin{align*}
		\Phi x_1=\Phi x_2.
	\end{align*}
	where $\Phi$ is as in (\ref{eq:def_phi_matrix}). In particular, $\Phi$ does not do $s$-sparse recovery for any $s\in\N$.
	\begin{proof}
		The vector $B^{\ast}\xi$ has either at least two zero entries or at least two nonzero entries, given that $n\geq 3$. Hence, there exist $1$-sparse vectors $x_1,x_2\in\C^n$ with $x_1\neq x_2$ such that \[\langle Bx_1,\xi\rangle=\langle x_1,B^{\ast}\xi\rangle=\langle x_2,B^{\ast}\xi\rangle=\langle Bx_2,\xi\rangle.\] By definition, this implies $\Phi x_1=\Phi x_2$.
	\end{proof}
\end{example}

The example above shows that recovery fails when choosing the trivial representation and that randomizing the generating vector $\xi$ or the sampling set $\Omega$ will be of no use. Intuitively, this makes sense because we do not gain any new information from new measurements. Conversely, we would wish for representations that connect the entries of the signal $x$ within each measurement in such a way that each new measurement provides more information about the signal.

\section{The restricted isometry property for a randomized generating vector}\label{sect:fix_zuf}
	
	In this section, we consider measurement matrices $\Phi$ where the sampling set $\Omega\subseteq G$ is fixed and chosen as in (\ref{item:fixed_sampling_sets}) and the generating vector $\xi$ is randomized. Further, the vector $x$ is  sparse in the standard basis, i.e. $B=I_n$. Thus, for a fixed sampling set $\Omega$ the measurement matrix is given as
		\begin{align}\label{eq:measurement_matrix_omega_fix_xi_random}
		\Phi=\frac{1}{\sqrt{m}} R_{\Omega}\big(\pi(g)\xi\big)_{g\in G}^{\ast}.
		\end{align}
		Here and for the remainder of Section \ref{sect:fix_zuf}, $\pi\colon G\to \text{GL}(\C^n)$ is a unitary representation or unitary projective representation unless specified differently. For reasons of readability we will often just speak of representations.\par		
		Our main result for this setting is Theorem \ref{mainresult1} which establishes the restricted isometry property for the measurement matrix in (\ref{eq:measurement_matrix_omega_fix_xi_random}). The proof of that result follows the same outline as the proof of \cite[Theorem 4.1]{krahmer2014suprema} and is presented in our conference paper \cite{fuhr2025restricted} in more detail. However, for self-containment we will present the main ideas of the proof.\par
		It is important to note that our result contains \cite[Theorem 4.1]{krahmer2014suprema}, which our proof is based on, as a special case. This will become more apparent when discussing the left regular representation in Section \ref{subsec:constant_analysis_ind_of_omega}.
		
	\begin{proof}[Proof of Theorem \ref{mainresult1}]
	Let $\Omega\in\widetilde{\Omega}$ be fixed. The main ingredient of the proof is the following concentration inequality.
	
	\begin{theorem}\label{th:concentration_inequ}\cite[Theorem 3.1]{krahmer2014suprema}
		Let $\mathcal{A}$ be a set matrices, and let $\xi$ be a random vector whose entries $\xi_j$ are independent, mean 0, variance 1 and $L$-subgaussian random variables. Define $d_{2\to 2}(\mathcal{A})=\sup_{A\in\mathcal{A}}\Vert A\Vert_{2\to 2}$ and $d_{F}(\mathcal{A})=\sup_{A\in\mathcal{A}}\Vert A\Vert_{F}$. By $\gamma_2(\mathcal{A},\Vert\cdot\Vert_{2\to2})$ we denote Talagrand's functional (see \cite[Definition 2.1]{krahmer2014suprema}). Set
		\begin{align*}
			E_1&=\gamma_2(\mathcal{A},\Vert\cdot\Vert_{2\to 2})(\gamma_2(\mathcal{A},\Vert\cdot\Vert_{2\to 2})+d_F(\mathcal{A}))+d_F(\mathcal{A})d_{2\to 2}(\mathcal{A}),\\
			E_2&=d_{2\to 2}(\mathcal{A})(\gamma_2(\mathcal{A},\Vert\cdot\Vert_{2\to 2})+d_F(\mathcal{A})),\\
			E_3&=d^2_{2\to 2}(\mathcal{A}).
		\end{align*}
		Then, for $t>0$,
		\begin{align*}
			&\mathbb{P}\left(\sup_{A\in\mathcal{A}}\left\vert\Vert A\xi\Vert_2^2-\mathbb{E}\Vert A\xi\Vert_2^2\right\vert \geq c_1E_1+t \right)\leq 2\exp\left(-c_2\min\left\{\frac{t^2}{E_2^2},\frac{t}{E_3}\right\}\right).
		\end{align*}
		The constants $c_1,c_2>0$ depend only on $L$.
	\end{theorem}
	
	Let us fix some notation
	\[
	D_{s,n}=\{x\in\C^n \mid \Vert x\Vert_2\leq 1, \Vert x\Vert_0\leq s\} \quad\text{and}\quad A_x=\frac{1}{\sqrt{m}} R_{\Omega}\left(\pi(g)^{\ast}x\right)^{\ast}_{g\in G}\in \C^{m \times n}
	\]
	for every $x\in\C^n$. The benefit of the concentration inequality gets more apparent when noticing that we can rewrite the restricted isometry constant as follows
	\begin{align*}
		\delta_s&=\sup_{x\in D_{s,n}} \Big| \Vert \Phi x \Vert_2^2-\Vert x\Vert^2_2 \Big|=\sup_{x\in D_{s,n}} \Big| \Vert A_x \xi\Vert_2^2-\mathbb{E} \Vert A_x \xi\Vert^2_2\Big|,
	\end{align*}
	where we used that the entries of $\xi$ are independent and have mean 0 and variance 1 as well as the fact that $\pi$ is unitary.
	So, Theorem \ref{th:concentration_inequ} tells us that it will be sufficient to find suitable bounds for $d_{2\to 2}(\mathcal{A})$, $d_{F}(\mathcal{A})$ and $\gamma_{2}(\mathcal{A},\Vert\cdot\Vert_{2\to 2})$, where $\mathcal{A}=\{ A_x \mid x\in D_{s,n}\}$.\par
	We start with $d_{2\to 2}(\mathcal{A})$: Let $x\in\C^n$ with $\Vert x\Vert_0\leq s$. There exist $J\subset\{1,\ldots,n\}$ and $(a_j)_{j\in J}\subseteq \C$ such that $x=\sum_{j\in J} a_j e_j$ with $\vert J\vert \leq s$. Then, the triangle inequality, assumption (\ref{main_annahme}) and the Cauchy-Schwarz inequality give
	\begin{align}\label{eq:2norm_bound_under_main_assumption}
		\Big\Vert \frac{1}{\sqrt{m}} R_{\Omega}(\pi(g)y)^{\!\ast}_{\!g\in G} x\Big\Vert_2&\leq \sum_{j \in J} \vert a_j\vert\Big\Vert \frac{1}{\sqrt{m}} R_{\Omega}(\pi(g)y)^{\!\ast}_{\!g\in G}  e_j\Big\Vert_2\nonumber\\
		& \leq \sum_{j\in J} \vert a_j\vert \sqrt{\frac{C_{\widetilde{\Omega},\pi}}{m}}\Vert y\Vert_2 \\
		& \leq \sqrt{\vert J\vert} \Big( \sum_{j\in J} \vert a_j\vert^2\Big)^{\frac{1}{2}}  \sqrt{\frac{C_{\widetilde{\Omega},\pi}}{m}} \Vert y\Vert_2\nonumber\\
		&\leq \sqrt{\frac{s\, C_{\widetilde{\Omega},\pi}}{m}} \Vert y\Vert_2 \Vert x\Vert_2\nonumber
	\end{align} 
	for all $y\in \C^n$. Since $\vert\langle y,\pi(g)^{\ast}x\rangle \vert=\vert \langle x,\pi(g)y\rangle \vert$ holds for all $y\in\C^n$ and $g\in\Omega$, it follows that
	\begin{align}\label{eq:Ax_2norm_bound}
		\Vert A_x\Vert_{2\to2} =\!\sup_{\Vert y\vert_2=1}\!\Big\Vert \frac{1}{\sqrt{m}}\, R_{\Omega}\big(\pi(g)^{\ast} x\big) ^{\ast}_{g\in G} y 	\Big\Vert_{2}=\!\sup_{\Vert y\vert_2=1}\!\Big\Vert \frac{1}{\sqrt{m}}\, R_{\Omega}\big(\pi(g) y\big) ^{\ast}_{g\in G} x 	\Big\Vert_{2}\leq \sqrt{\!\frac{s\, C_{\widetilde{\Omega},\pi}}{m}\!} \Vert x\Vert_2. 
	\end{align}
	Hence,
	\begin{align}\label{eq:bound_for_d_22}
		d_{2\to 2}(\mathcal{A})\leq \sqrt{\frac{s\, C_{\widetilde{\Omega},\pi}}{m}}.
	\end{align}

	Next we consider the quantity $d_F(\mathcal{A})$: Let $x\in\C^n$ with $\Vert x\Vert_0\leq s$. The definition of the Frobenius norm, along with $\pi$ being unitary, implies
	\begin{align*}
		\Vert A_x \Vert_F^2&=\frac{1}{m} \sum_{g\in\Omega}\sum_{j=1}^n \vert (\pi(g)^{\ast}x)_j\vert^2=\frac{1}{m} \sum_{g\in \Omega} \Vert \pi(g)^{\ast} x\Vert_2^2 =\frac{1}{m} \sum_{g\in \Omega} \Vert x\Vert_2^2=\Vert x\Vert_2^2.
	\end{align*}
	Thus,
	\begin{align}\label{eq:bound_for_d_F}
	d_F(\mathcal{A})=1.
	\end{align}

	It remains to bound the term $\gamma_{2}(\mathcal{A},\Vert\cdot\Vert_{2\to 2})$: In \cite{dudley1967sizes}, Dudley was the first to prove so called Dudley integral bounds. We will use a version of these bounds that was explicitly shown in \cite{vershynin2018high}. It reads as
	\begin{align}\label{eq:dudley_in}
		\gamma_{2}(\mathcal{A},\Vert\cdot\Vert_{2\to 2})\lesssim \int\limits_{0}^{d_{2\to 2}(\mathcal{A})}\sqrt{\ln \mathcal{N}(\mathcal{A},\Vert\cdot\Vert_{2\to 2},t)}\,\text{d}t
	\end{align}
	where $\mathcal{N}(\mathcal{A},\Vert\cdot\Vert_{2\to 2},t)$ is the covering number of $\mathcal{A}$, i.e. the minimal number of open balls in $(\mathcal{A},\Vert\cdot\Vert_{2\to 2}$) of radius $t$ that is needed to cover the set $\mathcal{A}$. We will split the integral and prove two bounds: The first bound uses the empirical method of Maurey \cite{carl1985inequalities} and the second one a volumetric argument. This approach has also been used in several similar settings, see \cite{rudelson2008sparse} or \cite{rauhut2010compressive}.\par
	Let $t>0$. Define a set of matrices by
	\begin{align*}
		\mathcal{B}=\{A_{\pm\sqrt{2}e_1},\ldots,A_{\pm\sqrt{2}e_n},A_{\pm\sqrt{2}i e_1},\ldots,A_{\pm\sqrt{2}i e_n} \}.
	\end{align*}
	The definition of $D_{s,n}$ implies
	\begin{align*}
		D_{s,n}\subseteq\! \sqrt{s}\,\text{conv}(\pm\sqrt{2}e_1,\ldots,\pm\sqrt{2}e_n,\pm\sqrt{2}ie_1,\ldots,\pm\sqrt{2}ie_n),
	\end{align*}
	where $\text{conv}()$ denotes the convex hull of a set. This relation together with
	\begin{align*}
		A_{\sum_{j=1}^Na_j z^j}=\sum_{j=1}^N \overline{a_j}\, A_{z^j}
	\end{align*}
	for all $a_1,\ldots,a_N\in\C$ and $z^1,\ldots,z^N\in\C^n$, shows $\mathcal{A}\subset \sqrt{s}\,\text{conv}(\mathcal{B})$. Thus, we have an upper bound on the covering number
	\begin{align}\label{eq:cov_numb_conv_subset}
		\mathcal{N}(\mathcal{A},\Vert\cdot\Vert_{2\to2},t)\leq \mathcal{N}\left(\text{conv}(\mathcal{B}),\Vert\cdot\Vert_{2\to2},\frac{1}{\sqrt{s}}t\right).
	\end{align}
	Now, we want to use \cite[Lemma 4.2]{krahmer2014suprema} which is based on the empirical method of Maurey. Let $N\in\N$, $(A_1,\ldots,A_N)\in \mathcal{B}^L$ and  let $(e_j)_{j=1}^N$ be a random vector with independent Rademacher distributed entries. One requirement of \cite[Lemma 4.2]{krahmer2014suprema} is a bound of the following form
	\begin{align*}
		\mathbb{E}_{\epsilon}\Big\Vert \sum_{j=1}^N \epsilon_j A_j\Big\Vert_{2\to 2}
		&\,\,{\lesssim} \sqrt{\ln(n)} \max\Big\{\Big\Vert \sum_{j=1}^N A_{j}^{\ast} A_{j}\Big\Vert_{2\to 2},\Big\Vert \sum_{j=1}^N A_{j} A_{j}^{\ast}\Big\Vert_{2\to 2}\Big\}^{\frac{1}{2}}\,\,{\leq} \sqrt{\ln(n)} \Big(\sum_{j=1}^N \Vert A_{j}\Vert^2_{2\to 2}\Big)^{\frac{1}{2}},
	\end{align*}
	where we used the non-commutative Khintchine inequality due to Lust-Piquard \cite{lust1986inegalites}, \cite{rudelson1999random} for the first inequality. Our assumption in (\ref{main_annahme}) implied the inequality in (\ref{eq:Ax_2norm_bound}), which in turn leads to
	\begin{align*}
		\sqrt{\ln(n)} \bigg(\sum_{j=1}^N \Vert A_{j}\Vert^2_{2\to 2}\bigg)^{\frac{1}{2}}\leq \sqrt{\ln(n)} \,\frac{\sqrt{2\,C_{\widetilde{\Omega},\pi}}}{\sqrt{m}}\,\sqrt{N}.
	\end{align*}
	Now it follows from \cite[Lemma 4.2]{krahmer2014suprema} that 
	\begin{align}\label{eq:cov_numb_maurey_bound}
		\!\!\ln \mathcal{N}\Big(\!\text{conv}(B) ,\Vert \!\cdot\!\Vert_{2\to 2}, \!\frac{t}{\sqrt{s}}\Big)\!\lesssim  \! \frac{{s\, C_{\widetilde{\Omega},\pi}}}{m} \!\frac{1}{t^2}{\ln (n)}\! \ln(4n).
	\end{align}
	Putting (\ref{eq:cov_numb_conv_subset}) and (\ref{eq:cov_numb_maurey_bound}) together yields
	\begin{align}\label{eq:cov_numb_ineq_for_big_t}
		\ln \mathcal{N}(\mathcal{A},\Vert\cdot\Vert_{2\to2},t)\lesssim \frac{{s\, C_{\widetilde{\Omega},\pi}}}{m} \frac{1}{t^2}{\ln (n)} \ln(4n)
	\end{align}
	The second inequality we aim to prove is based on an volumetric arguement. Inspecting (\ref{eq:2norm_bound_under_main_assumption}) gives us
	\begin{align*}
		\Vert A_x-A_y\Vert_{2\to 2}=\Vert A_{x-y}\Vert_{2\to 2}\leq \sqrt{\frac{ C_{\widetilde{\Omega},\pi}}{m}}\Vert x-y\Vert_1.
	\end{align*}
	Hence, 
	\begin{align*}
		\mathcal{N}(\mathcal{A},\Vert\cdot\Vert_{2\to2},t) &\leq \mathcal{N}\left( D_{s,n},\sqrt{\frac{C_{\widetilde{\Omega},\pi}}{m}} \Vert\cdot\Vert_1,t \right).
	\end{align*}
	Following the arguments presented in \cite[Section 8.4]{rauhut2010compressive}, it can be concluded that
	\begin{align*}
		\mathcal{N}\left(\! D_{s,n},\sqrt{\frac{C_{\widetilde{\Omega},\pi}}{m}} \Vert\cdot\Vert_1,t \right)\leq \left(\!1+2\sqrt{\frac{s\, C_{\widetilde{\Omega},\pi}}{m}}\frac{1}{t}\right)^{2s} \!\left(\frac{en}{s}\right)^s.
		\end{align*}
	Thus, we have shown another bound on the covering number of $\mathcal{A}$,
	\begin{align}\label{eq:cov_numb_ineq_for_small_t}
		\ln\mathcal{N}(\mathcal{A},\Vert\cdot\Vert_{2\to2},t)\lesssim s\Bigg(\ln  \Bigg(1+2\sqrt{\frac{s\, C_{\widetilde{\Omega},\pi}}{m}}\frac{1}{t}\Bigg)
		& +\ln\left(\frac{en}{s}\right)\Bigg).
	\end{align}	
	
	Let's prove the bound for $\gamma_2(\mathcal{A},\Vert\cdot\Vert_{2\to2})$. Using  (\ref{eq:cov_numb_ineq_for_big_t}) gives
	\begin{align*}
		\int\limits_{\frac{1}{\sqrt{m}}}^{d_{2\to2}(\mathcal{A})} \!\!\!\!\sqrt{\ln \mathcal{N}(\mathcal{A},\Vert\cdot\Vert_{2\to2},t)} \,\,\text{d}t&\lesssim \int\limits_{\frac{1}{\sqrt{m}}}^{\sqrt{\frac{s\, C_{\widetilde{\Omega},\pi}}{m}}} \!\!\!\!\sqrt{C_{\widetilde{\Omega},\pi}} \sqrt{\frac{s}{m}} \sqrt{\ln (n )} \sqrt{\ln (4n )} \,\,\frac{1}{t} \,\,\text{d}t\\
		&=\frac{1}{2}\sqrt{C_{\widetilde{\Omega},\pi}} \sqrt{\frac{s}{m}} \sqrt{\ln (n )}\sqrt{\ln (4n )}\ln(s\, C_{\widetilde{\Omega},\pi}).
	\end{align*}
	For small values of $t$ we use the bound in (\ref{eq:cov_numb_ineq_for_small_t}) to obtain
	\begin{align*}
		\int\limits_{0}^{\frac{1}{\sqrt{m}}} \sqrt{\ln(\mathcal{N}(\mathcal{A}, \Vert \cdot \Vert_{2 \to 2}, t))} \,\,\text{d}t&\lesssim \sqrt{s} \int\limits_0^{\frac{1}{\sqrt{m}}}\sqrt{ \ln\Big(\frac{en}{s}\Big)+\ln\Big(1+2\frac{\sqrt{s\, C_{\widetilde{\Omega},\pi}}}{\sqrt{m}\,t}\Big)}  \,\,\text{d}t\\
		&\leq\sqrt{\frac{s}{m}} \bigg( \sqrt{\ln\Big(\frac{en}{s}\Big)}+\sqrt{\ln\Big(e\big(1+2\sqrt{s\, C_{\widetilde{\Omega},\pi}}\big)\Big)}\,\bigg),
	\end{align*}
	with the second inequality resulting from a similar calculation as presented in \cite[Section 8.4]{rauhut2010compressive}. Using Dudley's inequality (\ref{eq:dudley_in}) and the two inequalities above implies
	\begin{align}\label{eq:bound_for_gamma_2}
			\gamma_2(\mathcal{A},\Vert\cdot\Vert_{2\to2}) \lesssim \sqrt{C_{\widetilde{\Omega},\pi}} \sqrt{\frac{s}{m}} \sqrt{\ln (n )\ln (4n )}\ln(s\, C_{\widetilde{\Omega},\pi}).
	\end{align}

	Now, Theorem \ref{mainresult1} is an immediate consequence of Theorem \ref{th:concentration_inequ} and the bounds (\ref{eq:bound_for_d_22}), (\ref{eq:bound_for_d_F}) and (\ref{eq:bound_for_gamma_2}):
	Let $\delta\in (0,1)$. Our proven bounds together with inequality (\ref{gl:absm}) yield
	\begin{align*}
		E_1= \gamma_{2}(\mathcal{A},\Vert\cdot\Vert_{2\to 2})^2+\gamma_{2}(\mathcal{A},\Vert\cdot\Vert_{2\to 2})+d_{2\to 2}(\mathcal{A}) \leq\frac{\delta^2}{c}+\frac{\delta}{\sqrt{c}}+\frac{\delta}{\sqrt{c}}\leq \frac{3\delta}{\sqrt{c}},
	\end{align*}
	where we assumed that $c\geq 1$. We get $E_1\lesssim \frac{\delta}{2c_1}$ where $c_1$ is the absolute constant from Theorem \ref{th:concentration_inequ} for a sufficiently large choice of $c$. The concentration inequality of Theorem \ref{th:concentration_inequ} implies 
	\begin{align*}
		\mathbb{P}(\delta_s\geq \delta)&\leq \mathbb{P}\Big(\delta_s\geq c_1 E_1+\frac{\delta}{2}\Big)\leq 2 \exp\left(-c_2 \min\left\{\frac{0.25\delta^2}{E_2^2}, \frac{0.5\delta}{E_3}\right\}\right).
	\end{align*}	
	Again using the bounds in (\ref{eq:bound_for_d_22}), (\ref{eq:bound_for_d_F}) and (\ref{eq:bound_for_gamma_2}) as well as inequality (\ref{gl:absm}) and the definition of $E_2$ and $E_3$ of Theorem \ref{th:concentration_inequ} gives the bound
	\begin{align*}
		&2 \exp\left(-c_2 \min\left\{\frac{0.25\delta^2}{E_2^2}, \frac{0.5\delta}{E_3}\right\}\right)\leq 2 \exp\Bigg(-c_2 c \frac{0.25}{\big(c^{-\frac{1}{2}}+1\big)^2} \ln(\eta^{-1})\Bigg)\leq 2 \exp\big(-\ln(\eta^{-1})\big)=2\eta.
	\end{align*}
	for $c$ sufficiently large.
	\end{proof}

	\subsection{Representations for which all sampling sets perform equally well} \label{subsec:constant_analysis_ind_of_omega}

		As already mentioned when discussing the main results, ideally we would choose $\widetilde{\Omega}=\mathcal{P}(G)$ while still getting a small constant $C_{\mathcal{P}(G),\pi}$. This section deals with this case. Hence, for the remainder of this section we will analyze and discuss the choice $\widetilde{\Omega}=\mathcal{P}(G)$. \par
		We start with one of the most important representations: the left regular representation $L$. For any finite group $G$ it is defined by $L\colon G\to\text{GL}(\C^G), g\mapsto L(g)$ with
		\begin{align*}
			L(g)\colon \C^G\to \C^G, \, (L(g)(f))(h):=f(g^{-1}h).
		\end{align*}
	
		To ensure that measurement matrix $\Phi_L$ associated with $L$ does $s$-sparse recovery with high probability, we have to bound the term 
		\begin{align*}
		\sup_{h\in G}\Big\Vert R_{\Omega}\big(L(g) y\big) ^{\ast}_{g\in G} e_h	\Big\Vert_{2}
		\end{align*}
		for every $y\in \C^G$ and $\Omega\in\mathcal{P}(G)$ which was first introduced in Section \ref{sect:fix_zuf}.
		
		\begin{proposition}\label{prop:leftregrep_const}
			Let $L$ be the left regular representation and fix an arbitrary sampling set $\Omega\in\mathcal{P}(G)$. Then, it holds
			\[\Big\Vert R_{\Omega}\big(L(g) y\big) ^{\ast}_{g\in G} e_h 	\Big\Vert_{2}\leq \Vert y\Vert_2\]
			for all $y\in\C^G$, all canonical vectors $e_h\in \C^G$.
			\begin{proof}
				For $y\in \C^G$ and a canonical vector $e_h$ we have
				\begin{align*}
					\Big\Vert R_{\Omega} \big(L(g)y\big)^{\ast}_{g\in G}\, e_h\Big\Vert_2^2 &\leq \Big\Vert   \big(L(g)y\big)^{\ast}_{g\in G}\, e_h\Big\Vert_2^2= \sum_{g\in G} \vert \langle e_h,L(g)y\rangle\vert^2=\sum_{g\in G} \vert y(g^{-1} h)\vert^2=\Vert y\Vert_2^2.
				\end{align*}
			\end{proof}
		\end{proposition}
	An immediate consequence of the statement of Proposition \ref{prop:leftregrep_const} is that we can choose
	\[
	C_{\mathcal{P}(G),L}=1
	\]
	for any finite group $G$. So in this case Theorem \ref{mainresult1} yields the good bound
	\begin{align*}
	m\gtrsim s\ln(s)^2\ln(n)^2
	\end{align*}
	on the number of measurements that is needed for the restricted isometry property to hold. This also answers the first question we asked when discussing our main results in Section \ref{sec:overview_results}. Most of the results in Sections \ref{subsec:constant_analysis_ind_of_omega} and \ref{subsec:constant_analysis_specific_omega} will be presented in a manner similar to Proposition \ref{prop:leftregrep_const}. The connection between these results and the constant $C_{\widetilde{\Omega},\pi}$ as well as the RIP of the measurement matrix associated with $\pi$ becomes apparent through a discussion similar to the one presented above.
		\begin{remark}
			The bound on the number of measurements we obtain when considering the left regular representation actually matches the best known bounds for partial random circulant matrices proven in \cite[Theorem 4.1]{krahmer2014suprema}. This should be no surprise since these matrices are a special case, obtained by choosing the group $\Z/n\Z$ along with the left regular representation of that group. Then,
		\begin{equation*}
			\big((L(j)\xi)_{j\in\Z/n\Z}^{\ast}\, x\big)_k=\langle x, L(k)\xi\rangle=\sum_{j=1}^n x(j) \overline{\xi(-k+j)}=\sum_{j=1}^n x(j) \tilde{\xi}(k-j)
		\end{equation*}
		for all $k\in\Z/n\Z$. This is the circular convolution that is considered in \cite{krahmer2014suprema}. Here, it is not important that we have $\tilde{\xi}$ instead of $\xi$ since its entries are only complex conjugated and shuffled. So with $\xi$ fulfilling the requirements of Theorem \ref{mainresult1} also $\tilde{\xi}$ fulfills the requirements.
		\end{remark}

		A natural follow up question is: Do representations that are in some sense close to the left regular representation also yield $C_{\mathcal{P}(G),\pi}=1$ approximately? Recall that the left regular representation can be block-diagonalized into $d_{\rho_1} \rho_1\oplus \dots \oplus d_{\rho_r}\rho_r$ by the Fourier transformation \cite{terras1999fourier}. Here we have $\hat{G}=\{\rho_1,\ldots,\rho_r\}$. Our next result states that the constant $C_{\mathcal{P}(G),\pi}$ gets smaller, and therefore better, for representations (in a specific basis) that have a block-diagonal form which is similar to the one of $L$. On the other hand, it shows that irreducible representations are in this case not well-suited for sparse recovery.\par
		To present the statement precisely
		, we first need to establish some notation. Fix a unitary representation $\pi$ which is given in block diagonal form. Then, any vector $z\in\C^n$ can be decomposed into
		\begin{align*}
			z=\begin{pmatrix}
				z^{1,1}\\
				\vdots\\
				z^{1,m_{\pi}(\pi_1)}\\
				\vdots\\
				z^{T,1}\\
				\vdots\\
				z^{T,m_{\pi}(\pi_T)}
			\end{pmatrix},
		\end{align*}
		where $z^{\tau,\kappa}\in\C^{d_{\pi_{\tau}}}$ for all $1\leq \tau\leq T$ and $1\leq \kappa\leq  d_{\pi_{\tau}}$. Further, define the bijective mapping
		\[\alpha_{\pi}\colon \{(\tau,\kappa,\iota)\in\N^3\mid 1\leq \tau\leq T,\, 1\leq \kappa\leq m_{\pi}(\pi_{\tau}),\, 1\leq \iota\leq d_{\pi_{\tau}}\}\to \{1,\ldots,n\}\]
		by
		\begin{align*}
			\alpha_{\pi}(\tau,\kappa,\iota):=\left(\sum_{t=1}^{\tau-1} d_{\pi_t} m_{\pi}(\pi_t)\right)+(\kappa-1)d_{\pi_{\tau}}+\iota.
		\end{align*}
	
		To identify a suitable realization for a representation that enables sparse recovery, we require the discrete Fourier transform matrix. In dimension $n$, this unitary matrix is defined by $\text{DFT}^n_{jk}:=\frac{1}{\sqrt{n}}e^{\frac{2\pi ijk}{n}}$ for all $j,k\in\{1,\ldots,n\}$.
		
		\begin{proposition}\label{prop:C_for_subrep_of_left_reg}
			Let $\pi$ be a unitary representation given in block diagonal form. Define a unitary matrix $U\in\C^{n\times n}$ by
			\begin{equation*}
				U=\textnormal{DFT}^n \cdot D
			\end{equation*}
			where $D=\textnormal{diag}(d_1,\ldots,d_n)\in\C^{n\times n}$ is a diagonal matrix with diagonal elements \[d_{\alpha_{\pi}(\tau,\kappa,\iota)}=\sqrt{d_{\pi_{\tau}}}\textnormal{DFT}^{d_{\pi_{\tau}}}_{\kappa \!\!\!\mod d_{\pi_{\tau}},\iota}\] for all $1\leq \tau\leq T$, $1\leq \kappa\leq m_{\pi}(\pi_{\tau})$ and $1\leq \iota \leq d_{\pi_{\tau}}$.
			Now choose an arbitrary sampling set $\Omega\in\mathcal{P}(G)$. Then, it holds
			\begin{align*}
				\Big\Vert R_{\Omega}\big(U\pi(g)U^{\ast} y\big) ^{\ast}_{g\in G} e_j 	\Big\Vert_{2} \leq \sqrt{\frac{\vert G\vert}{n }}\sqrt{\max_{1\leq \tau\leq T} \left\lceil \frac{m_{\pi}(\pi_{\tau})}{d_{\pi_{\tau}}} \right\rceil} \Vert y\Vert_2
			\end{align*}
			for all $y\in\C^n$ and all canonical vectors $e_j\in \C^n$.
			\begin{proof}	
				Let's begin by establishing some key properties of the matrix $U$. As the product of two unitary matrices, $U$ itself is unitary. Fix an index $1\leq j\leq n$ and an irreducible representation by $1\leq \tau\leq T$. For $1\leq \kappa_1,\kappa_2 \leq m_{\pi}(\pi_{\tau})$ we get
				\begin{align*}
					\left\langle \begin{pmatrix}
						U_{j,\alpha_{\pi}(\tau,\kappa_1,1)}\\
						\vdots\\
						U_{j,\alpha_{\pi}(\tau,\kappa_1,d_{\pi_{\tau}})}\\
					\end{pmatrix}, \begin{pmatrix}
					U_{j,\alpha_{\pi}(\tau,\kappa_2,1)}\\
					\vdots\\
					U_{j,\alpha_{\pi}(\tau,\kappa_2,d_{\pi_{\tau}})}\\
					\end{pmatrix}\right\rangle & =\frac{1}{n}\left\langle \begin{pmatrix}					
					e^{\frac{2\pi i j \alpha_{\pi}(\tau,\kappa_1,1)}{n}} e^{\frac{2\pi i \kappa_1 1}{d_{\pi_{\tau}}}}\\
					\vdots\\
					e^{\frac{2\pi i j \alpha_{\pi}(\tau,\kappa_1,d_{\pi_{\tau}})}{n}} e^{\frac{2\pi i \kappa_1 d_{\pi_{\tau}}}{d_{\pi_{\tau}}}}\\
					\end{pmatrix}, \begin{pmatrix}
					e^{\frac{2\pi i j \alpha_{\pi}(\tau,\kappa_2,1)}{n}} e^{\frac{2\pi i \kappa_2 1}{d_{\pi_{\tau}}}}\\
					\vdots\\
					e^{\frac{2\pi i j \alpha_{\pi}(\tau,\kappa_2,d_{\pi_{\tau}})}{n}} e^{\frac{2\pi i \kappa_2 d_{\pi_{\tau}}}{d_{\pi_{\tau}}}}\\
					\end{pmatrix}\right\rangle\\
					&=\frac{1}{n}e^{\frac{2\pi ij\left( \alpha_{\pi}(\tau,\kappa_1,1)-\alpha_{\pi}(\tau,\kappa_2,1)\right)}{n}} \sum_{\iota=1}^{d_{\pi_{\tau}}} e^{\frac{2\pi i j(\iota-1)}{n}} e^{\frac{2\pi i \kappa_1 \iota}{d_{\pi_{\tau}}}} e^{-\frac{2\pi i j(\iota-1)}{n}} e^{-\frac{2\pi i \kappa_2 \iota}{d_{\pi_{\tau}}}}\\
					&=\frac{1}{n} e^{\frac{2\pi ij\left( \alpha_{\pi}(\tau,\kappa_1,1)-\alpha_{\pi}(\tau,\kappa_2,1)\right)}{n}} \sum_{\iota=1}^{d_{\pi_{\tau}}} e^{\frac{2\pi i (\kappa_1-\kappa_2) \iota}{d_{\pi_{\tau}}}}
				\end{align*}
				where we used that $\alpha_{\pi}(\tau,\kappa,\iota)=\alpha_{\pi}(\tau,\kappa,1)+\iota-1$ holds. We use the known equality
				\begin{align*}
					\sum_{\iota=1}^{d_{\pi_{\tau}}} e^{\frac{2\pi i (\kappa_1-\kappa_2) \iota}{d_{\pi_{\tau}}}}=\begin{cases}
						d_{\pi_{\tau}},&\kappa_1=\kappa_2 \!\!\mod d_{\pi_{\tau}},\\
						0,& \text{otherwise},
					\end{cases}
				\end{align*}
				to obtain
				\begin{align}\label{eq:subrows_of_U_orthogonal}
					\left\langle \begin{pmatrix}
						U_{j,\alpha_{\pi}(\tau,\kappa_1,1)}\\
						\vdots\\
						U_{j,\alpha_{\pi}(\tau,\kappa_1,d_{\pi_{\tau}})}\\
					\end{pmatrix}, \begin{pmatrix}
						U_{j,\alpha_{\pi}(\tau,\kappa_2,1)}\\
						\vdots\\
						U_{j,\alpha_{\pi}(\tau,\kappa_2,d_{\pi_{\tau}})}\\
					\end{pmatrix}\right\rangle=\begin{cases}
					e^{\frac{2\pi ij\left( \alpha_{\pi}(\tau,\kappa_1,1)-\alpha_{\pi}(\tau,\kappa_2,1)\right)}{n}} \frac{d_{\pi_{\tau}}}{n},&\kappa_1=\kappa_2 \!\!\mod d_{\pi_{\tau}},\\
					0,& \text{otherwise}.
				\end{cases}
				\end{align}
				Now we return to our desired statement. Hence, fix some vector $y\in \C^n$ and index $1\leq j\leq n$. For $g\in G$ we can rewrite the inner product $\langle e_j, U \pi(g) U^{\ast}y\rangle$ in the following way
				\begin{align*}
					\langle e_j, U \pi(g) U^{\ast}y\rangle&=\langle e_j, \sum_{l=1}^n \big(\pi(g)U^{\ast}y\big)_l \, U_{-,l}\rangle=\sum_{l=1}^n \overline{ \big(\pi(g)U^{\ast}y\big)_l}\, U_{j,l}=\sum_{\tau=1}^T\! \sum_{\kappa=1}^{m_\pi(\pi_{\tau})} \!\sum_{\iota=1}^{d_{\pi_{\tau}}} \overline{\big(\pi_{\tau}(g)(U^{\ast}y)^{\tau,\kappa}\big)_{\iota}} \,  U_{j,\alpha_{\pi}(\tau,\kappa,\iota)},
				\end{align*}
				Thus,
				\begin{align}\label{eq:rewrite_as_trace}
					\sum_{\tau=1}^T \sum_{\kappa=1}^{m_\pi(\pi_{\tau})} \sum_{\iota=1}^{d_{\pi_{\tau}}} \overline{\big(\pi_{\tau}(g)(U^{\ast}y)^{\tau,\kappa}\big)_{\iota}} \,  U_{j,\alpha_{\pi}(\tau,\kappa,\iota)} &=\sum_{\tau=1}^T \sum_{\kappa=1}^{m_\pi(\pi_{\tau})} \sum_{\iota_1=1}^{d_{\pi_{\tau}}} \sum_{\iota_2=1}^{d_{\pi_{\tau}}} \overline{\pi_{\tau}(g)_{\iota_1,\iota_2}}\, \overline{\big((U^{\ast}y)^{\tau,\kappa}\big)_{\iota_2}} \,  U_{j,\alpha_{\pi}(\tau,\kappa,\iota_1)}\nonumber\\
					&=\frac{\vert G\vert}{\vert G\vert}\sum_{\tau=1}^T d_{\pi_{\tau}} \text{tr}\Bigg(\sum_{\kappa=1}^{m_\pi(\pi_{\tau})} \frac{1}{d_{\pi_{\tau}}}\underbrace{\begin{pmatrix}
						U_{j,\alpha_{\pi}(\tau,\kappa,1)}\\
						\vdots\\
						U_{j,\alpha_{\pi}(\tau,\kappa,d_{\pi_{\tau}})}\\
					\end{pmatrix} \big((U^{\ast} y)^{\tau,\kappa}\big)^{\ast}}_{=A^{\tau,\kappa}} \pi_{\tau}(g)^{\ast}\Bigg).
				\end{align}
				Define $B\in\bigoplus_{\rho\in\hat{G}} \text{End}(\C^{d_{\rho}})$ by
				\[
				B(\rho)z:=\begin{cases}
					\sum_{\kappa=1 }^{m_{\pi}(\rho)}\frac{1}{d_{\pi_{\tau}}} A^{\tau,\kappa}z, & \exists \tau\in\{1,\ldots,T\}:\pi_{\tau}=\rho,\\
					0,& \text{otherwise},\\
				\end{cases}
				\]
				for every $\rho\in\hat{G}$ and $z\in \C^{d_{\rho}}$. With this notation at hand we can rewrite the last term in (\ref{eq:rewrite_as_trace}) as
				\begin{align*}
					\frac{\vert G\vert}{\vert G\vert}\sum_{\rho\in\hat{G}} d_{\rho} \text{tr}\left( B(\rho)\cdot \rho(g^{-1})\right)
					=\vert G\vert \cdot \mathcal{F}^{-1}(B)(g)
				\end{align*}
				where $\text{tr}$ now means the trace of an endomorphism. Plancherel's Theorem gives
				\begin{align*}
					\left\Vert \big( \langle e_j , U\pi(g) U^{\ast}y\rangle \big)_{g\in G}\right\Vert_2^2&=\vert G \vert^2\Vert \mathcal{F}^{-1}(B)\Vert_2^2=\vert G\vert^2\frac{1}{\vert G\vert} \sum_{\rho\in\hat{G}} d_{\rho} \text{tr}\left((\mathcal{F}(\mathcal{F}^{-1}B))(\rho) (\mathcal{F}(\mathcal{F}^{-1}B))(\rho)^{\ast} \right)\\
					&=\vert G\vert \sum_{\rho\in\hat{G}} d_{\rho} \text{tr}\left(B(\rho) (B(\rho))^{\ast} \right).
				\end{align*}
				Again identifying $\text{End}(V_{\rho})$ with $\C^{d_{\rho}\times\rho}$ via the transformation matrices with respect to the standard basis gives $\text{tr}(B(\rho)B(\rho)^{\ast})=\Vert B(\rho)\Vert_F^2$ and therefore,
				\begin{align*}
					\vert G\vert \sum_{\rho\in\hat{G}} d_{\rho} \text{tr}\left(B(\rho) (B)(\rho)^{\ast} \right)&=\vert G\vert \sum_{\tau=1}^T d_{\pi_{\tau}}\left\Vert \sum_{\kappa=1}^{m_{\pi}(\pi_{\tau})} \frac{1}{d_{\pi_{\tau}}} A^{\tau,\kappa}\right\Vert_F^2=\vert G\vert \sum_{\tau=1}^T \frac{1}{d_{\pi_{\tau}}}\left\Vert \sum_{\kappa=1}^{m_{\pi}(\pi_{\tau})}  A^{\tau,\kappa}\right\Vert_F^2.
				\end{align*}
				It remains to bound the Frobenius norm in the sum above. We define index sets by
				\[
				J_1=\{1,\ldots,d_{\pi_{\tau}}\},\quad J_2=\{d_{\pi_{\tau}}+1,\ldots,2d_{\pi_{\tau}} \},\quad\ldots\quad,J_{\left\lceil \frac{m_{\pi}(\pi_{\tau})}{d_{\pi_{\tau}}} \right\rceil} =\left\{\left(\left\lceil \frac{m_{\pi}(\pi_{\tau})}{d_{\pi_{\tau}}} \right\rceil-1\right)  d_{\pi_{\tau}}+1,\ldots, m_{\pi}(\pi_{\tau})\right\}
				\]
				where $\lceil \cdot\rceil$ denotes the ceiling function. The triangle inequality together with the standard $1$-norm bound gives
				\begin{align*}
					\left\Vert \sum_{\kappa=1}^{m_{\pi}(\pi_{\tau})}  A^{\tau,\kappa}\right\Vert_F^2\leq \left(\sum_{l=1}^{\left\lceil \frac{m_{\pi}(\pi_{\tau})}{d_{\pi_{\tau}}} \right\rceil} \left\Vert \sum_{\kappa\in J_l} A^{\tau,\kappa}\right\Vert_F\right)^2 \leq \left\lceil \frac{m_{\pi}(\pi_{\tau})}{d_{\pi_{\tau}}} \right\rceil \sum_{l=1}^{\left\lceil \frac{m_{\pi}(\pi_{\tau})}{d_{\pi_{\tau}}} \right\rceil
					} \left\Vert \sum_{\kappa\in J_l} A^{\tau,\kappa}\right\Vert_F^2 .
				\end{align*}
				We know that for all $1\leq l\leq \left\lceil \frac{m_{\pi}(\pi_{\tau})}{d_{\pi_{\tau}}} \right\rceil$ and all indices $\kappa_1,\kappa_2\in J_l$ with $\kappa_1\neq \kappa_2$ it holds $\kappa_1\neq \kappa_2 \!
				\mod d_{\pi_{\tau}}$. Therefore, recalling the definition of $A^{\tau,\kappa}$ as well as property (\ref{eq:subrows_of_U_orthogonal}) of the matrix $U$ gives
				\begin{align*}
					\left\Vert \sum_{\kappa\in J_l} A^{\tau,\kappa}\right\Vert_F^2& =\sum_{\kappa_1,\kappa_2\in J_l}\langle \begin{pmatrix}
						U_{j,\alpha_{\pi}(\tau,\kappa_1,1)}\\
						\vdots\\
						U_{j,\alpha_{\pi}(\tau,\kappa_1,d_{\pi_{\tau}})}\\
					\end{pmatrix}, \begin{pmatrix}
					U_{j,\alpha_{\pi}(\tau,\kappa_2,1)}\\
					\vdots\\
					U_{j,\alpha_{\pi}(\tau,\kappa_2,d_{\pi_{\tau}})}\\
				\end{pmatrix}\rangle \langle (U^{\ast} y)^{\tau,\kappa_2}, (U^{\ast} y)^{\tau,\kappa_1}\rangle=\sum_{\kappa\in J_l} \frac{d_{\pi_{\tau}}}{n}\Vert (U^{\ast}y)^{\tau,\kappa}\Vert_2^2.
				\end{align*}
				Combining the two equations mentioned above results in
				\begin{align*}
					\left\Vert \sum_{\kappa=1}^{m_{\pi}(\pi_{\tau})}  A^{\tau,\kappa}\right\Vert_F^2\leq \left\lceil \frac{m_{\pi}(\pi_{\tau})}{d_{\pi_{\tau}}} \right\rceil \frac{d_{\pi_{\tau}}}{n} \sum_{l=1}^{\left\lceil \frac{m_{\pi}(\pi_{\tau})}{d_{\pi_{\tau}}} \right\rceil}  \sum_{\kappa\in J_l} \Vert (U^{\ast}y)^{\tau,\kappa}\Vert_2^2=\left\lceil \frac{m_{\pi}(\pi_{\tau})}{d_{\pi_{\tau}}} \right\rceil \frac{d_{\pi_{\tau}}}{n} \sum_{\kappa=1}^{m_{\pi}(\pi_{\tau})} \Vert (U^{\ast}y)^{\tau,\kappa}\Vert_2^2
				\end{align*}
				By combining all the preceding information, we conclude that
				\begin{align*}
					\left\Vert \big( \langle e_j , U\pi(g) U^{\ast}y\rangle \big)_{g\in G}\right\Vert_2^2&\leq \vert G\vert \sum_{\tau=1}^T \frac{1}{d_{\pi_{\tau}}}  \left\lceil \frac{m_{\pi}(\pi_{\tau})}{d_{\pi_{\tau}}} \right\rceil \frac{d_{\pi_{\tau}}}{n} \sum_{\kappa=1}^{m_{\pi}(\pi_{\tau})} \Vert (U^{\ast}y)^{\tau,\kappa}\Vert_2^2\\
					&=\frac{\vert G \vert}{n}  \max_{1\leq \tau\leq T} \left\lceil \frac{m_{\pi}(\pi_{\tau})}{d_{\pi_{\tau}}} \right\rceil \sum_{\tau=1}^T \sum_{\kappa=1}^{m_{\pi}(\pi_{\tau})} \Vert(U^{\ast} y)^{\tau,\kappa}\Vert_2^2\\
					&=\frac{\vert G \vert}{n} \max_{1\leq \tau\leq T} \left\lceil \frac{m_{\pi}(\pi_{\tau})}{d_{\pi_{\tau}}} \right\rceil  \Vert U^{\ast} y\Vert_2^2\\
					&=\frac{\vert G \vert}{n} \max_{1\leq \tau\leq T} \left\lceil \frac{m_{\pi}(\pi_{\tau})}{d_{\pi_{\tau}}} \right\rceil  \Vert y\Vert_2^2.
				\end{align*}
				This finishes the proof. 
			\end{proof}
		\end{proposition}
		Inspecting the above proof shows that one gets the following sharp bound for irreducible representations.
		\begin{corollary}
			Let $\pi\colon G\to \text{GL}(\C^{d_{\pi}})$ be an unitary irreducible representation and choose an arbitrary sampling set $\Omega\in\mathcal{P}(G)$. Then, it holds
			\begin{align*}
				\Big\Vert R_{G}\big(\pi(g) y\big) ^{\ast}_{g\in G} e_j 	\Big\Vert_{2}  = \sqrt{\frac{\vert G\vert}{d_{\pi} }} \Vert y\Vert_2\quad \text{and}\quad \Big\Vert R_{\Omega}\big(\pi(g) y\big) ^{\ast}_{g\in G} e_j 	\Big\Vert_{2}  \leq \sqrt{\frac{\vert G\vert}{d_{\pi} }} \Vert y\Vert_2
			\end{align*}
			for all $y\in\C^n$ and all canonical vectors $e_j\in \C^n$.
		\end{corollary}
	
		Since the aforementioned bound is sharp and it holds $\sum_{\pi\in\hat{G}}d_{\pi}^2=\vert G\vert$ \cite{terras1999fourier}, it follows that $\frac{\vert G\vert}{d_{\pi} }$ tends to be quite large. Consequently, in this context, choosing irreducible representations for the measurement process is generally not effective for sparse recovery.\par
		That we can't expect the same $s$-sparse recovery properties for the measurement matrices of unitarily equivalent representations, is intuitive, since sparsity is a property that is not invariant under basis transformation. This means that the choice of the realization of a representation is important. In the following we want to prove our intuition. Thereby, we also show the necessity of a suitable basis transformation $U$ in Proposition \ref{prop:C_for_subrep_of_left_reg}. \par
		First, we will need an auxiliary statement that is inspired by so called delta trains \cite{foucart2013mathematical}.
		
		\begin{lemma}\label{lem:delta_trains}
			Let $s\vert n$. Then, the vector $v\in\C^n$ defined by
			\begin{align*}
				v_j=\begin{cases}
					1, & j\equiv 1 \textnormal{ mod } \frac{n}{s},\\
					0, &\text{otherwise},
				\end{cases}
			\end{align*}
			satisfies $\Vert v\Vert_0=s$ and $\Vert \mathcal{F}v\Vert_0=\frac{n}{s}$. Furthermore, it holds \[\textnormal{supp}(\mathcal{F}v)=\{l\in\{1,\ldots,n\} \colon l\equiv 0 \textnormal{ mod } s\}.\]
			\begin{proof}
				It is obvious that $\Vert v\Vert_0=s$ holds. The Fourier transform of the vector $v$ is given by
				\begin{align*}
					(\mathcal{F}v)_l&=\sum_{j=1}^n v_j\, e^{\frac{2\pi i j l 	}{n}}=\sum_{j=1}^s  e^{\frac{2\pi i ((j-1)\frac{n}{s}+1) l }{n}}=e^{\frac{2\pi i l }{n}}\sum_{j=1}^s  e^{\frac{2\pi i (j-1)l }{s}}=\begin{cases}
						e^{\frac{2\pi i l }{n}} s, & l\equiv 0 \text{ mod } s,\\
						0, & \text{otherwise}.
					\end{cases}
				\end{align*}
				This proves the claim.	
			\end{proof}
		\end{lemma}
		
		We now return to our initial question.
		
		\begin{example}\label{bsp:F_faltung_F}
			Let $G=\Z/n\Z$ and consider the left regular representation $L$. Theorem \ref{mainresult1} together with Proposition \ref{prop:leftregrep_const} tells us that the measurement matrix $\Phi_L$ associated with $L$ does $s$-sparse recovery with high probability for any sampling set $\Omega\in\mathcal{P}(G)$. Now, define a unitarily equivalent representation by
			\begin{equation*}
				\rho(k):=\left(\frac{\mathcal{F}}{\sqrt{n}}\right)\circ L(k) \circ \left(\frac{\mathcal{F}}{\sqrt{n}}\right)^{-1}=\mathcal{F}\circ L(k) \circ \mathcal{F}^{-1}
			\end{equation*}
			for all $k\in \Z/n\Z$. Hence, 
			\begin{align*}
				\left(\mathcal{F}\circ L(k) \circ \mathcal{F}^{-1} \right)(x)_l&=\frac{1}{n} \sum_{j_1=1}^n \sum_{j_2=1}^n x_{j_2} \,e^{-\frac{2\pi i j_2 (j_1-k) }{n}}\, 	e^{\frac{2\pi i j_1 l }{n}}=\frac{1}{n} \sum_{j_2=1}^n x_{j_2} \,e^{\frac{2\pi i j_2 k}{n}} \sum_{j_1=1}^n \, 	e^{-\frac{2\pi i j_1 (j_2-l) }{n}}= x_{l} \,e^{\frac{2\pi i l k}{n}}.
			\end{align*}
			Let $s,n\geq 2$ with $s\vert n$. Then, there exists a sampling set $\Omega\subseteq G$ with $\vert \Omega\vert=n-\frac{n}{s}$ such that for all $\xi\in\C^n$ there exist $s$-sparse vectors $x_1,x_2\in\C^n$ with $x_1\neq x_2$ such that 
			\[\Phi_{\rho} x_1=\Phi_{\rho} x_2\] where $\Phi_{\rho}=\frac{1}{\sqrt{m}}\, R_{\Omega}\big(\rho(k) \xi\big) ^{\ast}_{k\in G}$. In particular, for all $\xi\in\C^n$, $\Phi_{\rho}$ does not do $s$-sparse recovery.
			\begin{proof}
				Define the sampling set $\Omega=\{k\in\{1,\ldots,n\} \colon k\not\equiv 0 \text{ mod } s\}$. It is enough to find a $s$-sparse vector $x\in\C^n$ such that $\Phi x=0$. We consider three different cases. If $\xi=0$ holds, then our statement is trivial. So now assume that $1\leq\Vert \xi\Vert_0\leq n-1$. Hence, there exists an index $j$ with $\xi_j=0$. Choose $x=e_j$. Then, it holds \[ \langle x,\rho(k)\xi\rangle=\sum_{l=1}^n x_l \overline{\xi_l}\, e^{-\frac{2\pi i l k}{n}}=0\] for all $k\in\Omega$. Thus, $\Phi x=0$. It remains to consider $\Vert \xi\Vert_0=n$. In this case, we choose $x\in \C^n$ by 
				\begin{align*}
					x_j=\begin{cases}
						\overline{\frac{1}{\xi_j}},& j\equiv 1 \text{ mod } \frac{n}{s},\\
						0, & \text{otherwise.}
					\end{cases}
				\end{align*}
				Since $\Vert \xi\Vert_0=n$, the vector is well defined and $s$-sparse. Set
				$v\in\C^n$ as $v_j=\overline{x_j} \xi_j$. Then, Lemma \ref{lem:delta_trains} shows that \[\text{supp}(\mathcal{F}v)=\{k\in\{1,\ldots,n\} \colon k\equiv 0 \text{ mod } s\}=\{1,\ldots,n\}\setminus \Omega.\]
				Thus, 
				\begin{align*}
					\langle x,\rho(k)\xi\rangle= \sum_{j=1}^n x_j \overline{\xi_j}\, e^{-\frac{2\pi i j k}{n}}=\overline{\sum_{j=1}^n \overline{x_j} \xi_j\, e^{\frac{2\pi i j k}{n}}}= \overline{(\mathcal{F}v)_k}=0
				\end{align*}
				for all $k\in\Omega$. Again, we get that $\Phi x=0$.
			\end{proof}
			This example shows that even $m=\vert \Omega\vert= n-\frac{n}{s}\geq \frac{n}{2}$ measurements are not enough in order to ensure that $\Phi_{\rho}$ does $s$-sparse recovery regardless of the choice of the generating vector $\xi$.
		\end{example}
		
		An implication of the above result is that sparse recovery from Fourier measurements is not possible for an arbitrary sampling set $\Omega\in\mathcal{P}(G)$ although there is an equivalent representation that allows $s$-sparse recovery with high probability for every sampling set $\Omega$. It also shows that if one considers any fixed sampling set $\Omega$ results of the form as in Theorem \ref{mainresult1} have to depend on the realisation of $\pi$. This is an important observation. It leads to the question of how we can change our measurement process such that the measurement matrices associated with equivalent representations have the same $s$-sparse recovery properties. We will answer this in section \ref{sect:random_xi_and_omega}. The key change is that we randomize the sampling set.
		
	\subsection{A priori restricted sampling sets}\label{subsec:constant_analysis_specific_omega}

		To motivate the contents of this section, we start with an example. Let $p\in\N$ be prime. We call the set $G_{\text{aff}}=\Z_p\times Z_p^{\ast}$ together with the operation
		\begin{align*}
			(k,l)(k^{\prime},l^{\prime}):=(k+lk^{\prime} \text{ mod } p, ll^{\prime}\text{ mod } p)
		\end{align*}
		the affine group. For the affine group we define a unitary representation $\rho\colon G_{\text{aff}}\to\text{GL}(\C^{p-1})$ by
		\begin{align}\label{eq:aff_group_rep}
			(\rho(k,l)y)(j)=e^{\frac{2\pi ijk}{p}}y(jl)
		\end{align}
		for all $(k,l)\in G_{\text{aff}}$ and $y\in\C^{p-1}$. We get the following bound.
		
		\begin{proposition}\label{prop:affine_group}
			Let $\rho$ be the representation defined in (\ref{eq:aff_group_rep}) and $\Omega\in \mathcal{P}(G)$ an arbitrary sampling set. Then, it holds
			\begin{align*}
				\Big\Vert R_{\Omega}\big(\rho(k,l) y\big) ^{\ast}_{(k,l)\in G_{\textnormal{aff}}} e_j 	\Big\Vert_{2} \leq \sqrt{\vert\Omega_1\vert} \Vert y\Vert_2
			\end{align*}
			for all $y\in\C^{p-1}$ and all canonical vectors $e_j\in \C^{p-1}$ with $\Omega_1:=\{k\in\Z_p\mid \exists l\in\Z_p^{\ast}\colon (k,l)\in \Omega\}$.
				\begin{proof}
				For $y\in \C^{p-1}$ and a canonical vector $e_j$ we have
				\begin{align*}
					\Big\Vert  R_{\Omega} \big(\rho(k,l)y\big)^{\ast}_{(k,l)\in G_{\textnormal{aff}}}\, e_j\Big\Vert_2^2\!=\!\!\sum_{(k,l)\in\Omega}\Big\vert\langle e_j,e^{\frac{2\pi i k\cdot}{p}}y(\cdot l)\rangle\Big\vert^2\leq\!\! \sum_{k\in\Omega_1}\sum_{l=1}^{p-1}\vert\langle e_j,y(\cdot l)\rangle\vert^2=\!\sum_{k\in\Omega_1}\sum_{l=1}^{p-1}\vert y(jl^{-1})\vert^2=\vert\Omega_1\vert \Vert y\Vert_2^2.
				\end{align*}
			\end{proof}
		\end{proposition}
		
		It is obvious that the size of $\Omega_1$ equals $p$ for some choices of sampling sets. Hence, the above Proposition gives $C_{\mathcal{P}(G),\rho}=p$ which implies a non-desirable bound on the number of measurements
		\[
		m\gtrsim sp\ln(p)^4.
		\] in Theorem \ref{mainresult1}. However, since the measurement matrix and especially the representation are known, we can argue that this knowledge can and should be used. By a priori restricting the set of possible sampling sets $\widetilde{\Omega}$ (which equaled $\mathcal{P}(G)$ in our discussion so far) to be a subset of e.g. $\{(1,l)\in G_{\text{aff}}\,\vert\, l\in\Z_p^{\ast}\}$, we get $\vert \Omega_1\vert=1$ in Proposition \ref{prop:affine_group} and $C_{\widetilde{\Omega},\rho}=1$. Then, according to Theorem \ref{mainresult1}, 
		sampling sets $\Omega\in\widetilde{\Omega}$ with at least 
		\[
		m\gtrsim s \ln(p)^2\ln(s)^2
		\]	
		elements provide $s$-sparse recovery with high probability.\par 	
		In this section, we generalize this idea of a priori restricting the sampling sets $\Omega$ to be elements of $\widetilde{\Omega}\subseteq \mathcal{P}(G)$, such that $C_{\widetilde{\Omega},\pi}=1$ holds. Hence, giving an answer to the second question we asked ourselves when discussing the main results in Section \ref{sec:overview_results}. This strategy works best for representations that are induced by some normal subgroup $H$ of $G$. In Appendix \ref{subsec:appendix_induced_rep}, we provide the necessary background for this section about induced representations.
		\begin{definition}
			Let $H$ be subgroup of $G$ and let $\sigma\colon H \to \text{GL}(W)$ be a representation of $H$. Define a vector space
			\begin{align*}
				V=\{f\colon G\to W\,\vert\, f(hg)=\sigma(h)f(g)\,\,\,\,\forall h\in H, g\in G\}.
			\end{align*}
			Further, define the group homomorphism $\textnormal{Ind}_{H}^G\sigma\colon G\to \text{GL}(V)$ by
			\begin{align*}
				(\textnormal{Ind}_H^G\sigma (g_1)f)(g_2)=f(g_2g_1)
			\end{align*}
			for all $g_1,g_2\in G$. $\textnormal{Ind}_{H}^G\sigma$ is called the induced representation from $H$ up to $G$. 
		\end{definition} 
		
		In the following we will only be interested in $W=\C^k$ and $\sigma$ being unitary. According to Appendix \ref{subsec:appendix_induced_rep} there exists a unitary map $T\colon \left(\C^k\right)^{H\setminus G}\to V$ such that
		\begin{align*}
			\big(T^{-1}(\textnormal{Ind}_H^G\sigma)(g) Tf\big)(\upsilon)=\sigma\left( \gamma(\upsilon)g \,\gamma(\upsilon g)^{-1} \right)f(\upsilon g)
		\end{align*}
		for all $g\in G$, $f\in\left(\C^k\right)^{H\setminus G}$ and $\upsilon\in H\setminus G$. Here, we denote the right cosets by $H\setminus G =\{Hg\mid g\in G\}$. Identify \[W^{H\setminus G}=\left(\C^k\right)^{H\setminus G}\cong \C^{\{1,\ldots,k\}\times H\setminus G}\cong\C^{k\cdot\vert H\setminus G\vert}\] via the standard identification map
		\begin{align*}
			\varphi\colon \C^{k\cdot\vert H\setminus G\vert}\to \left(\C^k\right)^{H\setminus G}.
		\end{align*}
		Now, we are ready to define the representation $\pi_{\sigma}\colon G\to \text{GL}\left(\C^{k\cdot\vert H\setminus G\vert}\right)$ by
		\begin{align}\label{eq:def_pi_sigma} 
			\pi_{\sigma}(g)= \varphi^{-1} T^{-1}(\textnormal{Ind}_H^G\sigma)(g) T \varphi.
		\end{align}
		It is obvious that $\pi_{\sigma}$ is equivalent to $\textnormal{Ind}_H^G \sigma $. Further, $\pi_{\sigma}$ is unitary since $\textnormal{Ind}_H^G \sigma$, $T$ and $\varphi$ are unitary. $\pi_{\sigma}$ is the representation we will prove a recovery result for. To do this, want to restrict the sampling set a priori. Recall that the cosets of $G$ define an equivalence relation on $G$. 
		\begin{lemma}
			Let $H$ be subgroup of $G$. Then, 
			\begin{align*}
				g_1\sim g_2\quad\Leftrightarrow\quad g_1 g_2^{-1}\in H
			\end{align*}
			is a well-defined equivalence relation on $G$.
		\end{lemma}
		
		 We want our sampling sets to contain at most one element of each coset. Thus, we set \[\widetilde{\Omega}=\{\Omega\subseteq G\,\vert\, \forall g_1,g_2\in \Omega : g_1 \not\sim g_2\}.\] With that notation the main result of this section reads as follows.
		
		\begin{proposition}
			Let $H$ be a normal subgroup of $G$ and let $\sigma\colon H\to \text{GL}(\C^k)$ be a unitary  representation of $H$. Further, let $\Omega\in\widetilde{\Omega}=\{\Omega\subseteq G\,\vert\, \forall g_1,g_2\in \Omega : g_1 \not\sim g_2\}$. Then, it holds
			\begin{align*}
				 \Big\Vert R_{\Omega}\big(\pi_{\sigma}(g) y\big) ^{\ast}_{(k,l)\in G_{\textnormal{aff}}} e_j 	\Big\Vert_{2}\leq  \Vert y\Vert_2
			\end{align*}
			for all $y\in \C^{k\cdot\vert H\setminus G\vert}$ and $j\in \C^{k\cdot\vert H\setminus G\vert}$, where $\pi_{\sigma}$ was defined in (\ref{eq:def_pi_sigma}).
			\begin{proof}
				Let $y\in \C^{k\cdot\vert H\setminus G\vert}$ and $j\in \C^{k\cdot\vert H\setminus G\vert}$. It holds
				\begin{align*}
					\sum_{g\in\Omega} \vert \langle e_{j},\pi(g)y\rangle\vert^2&=\sum_{g\in\Omega} \left\vert\langle \varphi(e_j), \sigma\left( \gamma(\cdot)g \,\gamma(\cdot g)^{-1} \right)\varphi(y)(\cdot g)\rangle_{\left(\C^k\right)^{H\setminus G}} \right\vert^2\\
					&=\sum_{g\in\Omega}\left\vert\sum_{\upsilon\in H\setminus G}\langle \varphi(e_j)(\upsilon),\sigma\left( \gamma(\upsilon)g \,\gamma(\upsilon g)^{-1} \right)\varphi(y)(\upsilon g)\rangle\right\vert^2\\
					&\leq \sum_{g\in\Omega}\left(\sum_{\upsilon\in H\setminus G}\Vert \varphi(e_j)(\upsilon)\Vert_2\left\Vert \sigma\left( \gamma(\upsilon)g \,\gamma(\upsilon g)^{-1} \right)\varphi(y)(\upsilon g)\right\Vert_2 \right)^2\\
					&=\sum_{g\in\Omega}\left(\sum_{\upsilon\in H\setminus G}\Vert \varphi(e_j)(\upsilon)\Vert_2\left\Vert \varphi(y)(\upsilon g)\right\Vert_2 \right)^2
				\end{align*}
				since $\sigma$ is unitary. Since we assumed $\varphi$ to be the natural identification map between $\left(\C^k\right)^{H\setminus G}$ and $\C^{k\cdot\vert H\setminus G\vert}$, we know that there exists exactly one $\nu\in H\setminus G$ such that $\Vert \varphi(e_j)(\upsilon)\Vert_2=\delta_{\upsilon,\nu}$ holds. Thus,
				\begin{align*}
					\sum_{g\in\Omega}\left(\sum_{\upsilon\in H\setminus G}\Vert \varphi(e_j)(\upsilon)\Vert_2\left\Vert \varphi(y)(\upsilon g)\right\Vert_2 \right)^2=\sum_{g\in\Omega}\Vert \varphi(y)(\nu g)\Vert_2^2.
				\end{align*}
				Since $H$ is normal, it holds
				\begin{align*}
					Hgg_1=Hgg_2 \quad\Leftrightarrow\quad H g g_1 g_2^{-1} g^{-1}=H \quad\Leftrightarrow\quad gg_1 g_2^{-1} g^{-1}\in H \quad\Leftrightarrow\quad g_1 g_2^{-1}\in H
				\end{align*}
				for all $g,g_1,g_2\in G$. Hence, $g_1, g_2\in \Omega$ with $g_1\neq g_2$ implies $\upsilon g_1\neq \upsilon g_2$. Thus,
				\begin{align*}
					\sum_{g\in\Omega} \left\Vert \varphi(y)(\nu g) \right\Vert_2^2\leq \sum_{\upsilon\in H\setminus G} \Vert \varphi(y)(\upsilon)\Vert_2^2=\Vert y\Vert_2^2.
				\end{align*}
				Then, the statement follows immediately.
			\end{proof}
		\end{proposition}
		We end this section by commenting on the possible number of choices for $\Omega$ and its maximum size.
		\begin{remark}
			The size of $\Omega\in\widetilde{\Omega}$ is bounded by the number of equivalence classes. Since $g_1\sim g_2$ is defined as $g_1 g_2^{-1}\in H$, we get $\vert [g]_{\sim}\vert=\vert H\vert$ for every $g\in G$. Thus, the number of equivalence classes is $\frac{\vert G\vert}{\vert H\vert}=\vert H\setminus G\vert$.\par
			With the previous considerations, we get that $\Omega\in \widetilde{\Omega}$ implies $\vert \Omega\vert \in \{0,1,\ldots,\vert H\setminus G\vert\}$. Fix an index $l\in\{1,\ldots,\vert H\setminus G\vert\}$\}. Then, there are
			\begin{align*}
				\frac{1}{l!}\prod_{j=0}^{l-1}(\vert G\vert- j\vert H\vert)=\frac{\vert H\vert^l}{l!}\prod_{j=0}^{l-1}\left(\frac{\vert G\vert}{\vert H\vert}- j\right)=\vert H\vert^l \binom{\frac{\vert G\vert}{\vert H\vert}}{l}
			\end{align*}
			choices for $\Omega\in \widetilde{\Omega}$ with $\vert \Omega\vert=l$. Therefore,
			\begin{align*}
				\vert \widetilde{\Omega}\vert=1+\sum_{l=1}^{\vert H\setminus G\vert} \vert H\vert^l \binom{\frac{\vert G\vert}{\vert H\vert}}{l}=\sum_{l=0}^{\vert H\setminus G\vert} \vert H\vert^l \binom{\vert H\setminus G\vert}{l}=(1+\vert H\vert)^{\vert H\setminus G\vert}
			\end{align*}
			where we used the binomial theorem.
		\end{remark}

\section{The restricted isometry property for a randomized generating vector and a randomized sampling set}\label{sect:random_xi_and_omega}

In this section we want to consider randomness in the generating vector $\xi$ as well as the sampling set $\Omega$. This is motivated by the following: Representation theory mostly studies properties of representations that are invariant under unitary basis transformation. Therefore, we would wish for a setting such that sparse recovery is invariant under unitary intertwining operators. However, example \ref{bsp:F_faltung_F} and the subsequent discussion showed that this can not be expected for a fixed sampling set $\Omega$. By randomizing $\Omega$, we hope to overcome this problem.
\subsection{Bounded orthonormal systems (BOS)}
We start by introducing some needed background about so called bounded orthonormal systems.

\begin{definition}[\cite{foucart2013mathematical}]
	Let $\mathcal{D}\subseteq \R^d$ be endowed with a probability measure $\mu$. Further, let $\{\varphi_1,\ldots,\varphi_n\}\subseteq \C^{\mathcal{D}}$ be an orthonormal system, i.e.
	\begin{equation}\label{eq:ortho_systems}
		\int\limits_{\mathcal{D}} \varphi_j(t)\overline{\varphi_{k}(t)}\,\text{d}\mu(t)=\delta_{j,k}
	\end{equation}
	for all $j,k\in\{1,\ldots,n\}$. We call $\{\varphi_1,\ldots,\varphi_n\}$ an bounded orthonormal system (BOS) with constant $K$ if it satisfies
	\begin{equation}\label{eq:BOS_cond}
		\max_{1\leq j\leq n}\sup_{t\in\mathcal{D}}\vert \varphi_j(t)\vert\leq K.
	\end{equation}
\end{definition}

The use of such systems for sparse recovery is well-established. To be more accurate, there exist results \cite{rudelson2008sparse,rauhut2010compressive} that establish e.g. the restricted isometry property for the measurement matrix
\[
\Phi=\big( \varphi_k(t_l) \big)_{1\leq l\leq m,1\leq k\leq n}
\]
where $t_1,\ldots, t_m$ are selected independently at random according to $\mu$.\par
Since we are working in a discrete setting the most important type of BOS arises from systems of orthogonal and normalized vectors.

\begin{example}\label{ex:discrete_OS}
	Let $U\in\C^{n\times N}$ be a matrix with orthonormal rows and let $B\in\C^{n\times n}$ be unitary. Now choose $\mathcal{D}=\{1,\ldots,N\}$ and the measure $\mu(A)=\frac{\vert A\vert}{N}$ for all $A\subseteq\{1,\ldots,N\}$. We claim that the set of columns of $\sqrt{N}U^{\ast}B$ denoted by
	\[\{\sqrt{N}v_1,\ldots, \sqrt{N}v_n\}\]
	satisfies (\ref{eq:ortho_systems}). To show this, we first note that  
	\begin{equation*}
		\langle v_j, v_k\rangle=\langle U^{\ast}Be_j, U^{\ast}Be_k\rangle=\langle B^{\ast}\underbrace{U U^{\ast}}_{=I_n}Be_j,e_k\rangle=\langle e_j,e_k\rangle=\delta_{j,k}
	\end{equation*}
	for all $j,k\in\{1,\ldots,n\}$. Thus,
	\begin{align*}
		\int\limits_{\{1,\ldots,N\}} \!\!\!\!\!\!\sqrt{N} v_j(t)\overline{\sqrt{N} v_{k}(t)}\,\text{d}\mu(t)=\sum_{t=1}^N N v_j(t)\overline{v_k(t)}\cdot \mu{\{t\}}=\sum_{t=1}^N  v_j(t)\overline{v_k(t)}=\langle v_j,v_k\rangle=\delta_{j,k}
	\end{align*}
	for all $j,k\in\{1,\ldots,n\}$. The boundness condition (\ref{eq:BOS_cond}) then reads as
	\begin{align*}
		\sqrt{N} \cdot\max_{1\leq j\leq n,1\leq t\leq N} \vert \langle U e_t,Be_j\rangle\vert=	\max_{1\leq j\leq n,1\leq t\leq N} \vert \sqrt{N} e_t^{\ast} U^{\ast} Be_j\vert=\max_{1\leq j\leq n,1\leq t\leq N} \vert \sqrt{N} v_j(t)\vert\leq K.
	\end{align*}
\end{example}

\subsection{The columns of the measurement matrix are likely to be a BOS}\label{subsec:likely_BOS}

We are interested in the case where the matrix $U$ from example \ref{ex:discrete_OS} is given as
\[U=\frac{1}{\sqrt{\vert G\vert}}\big(\pi(g)\xi\big)_{g\in G} \in\C^{n\times \vert G\vert}.\]
Here, we assume that $\pi\leq L$ is a unitary subrepresentation of the left regular representation that is given
in block-diagonal form 
as in (\ref{eq:block_diag_form_pi}). Further, the random vector $\xi\in\C^n$ is constructed as follows (for notation see proof of Proposition \ref{prop:C_for_subrep_of_left_reg}): Let $(\epsilon_{\tau,\iota})_{1\leq \tau\leq T,\, 1\leq \iota\leq d_{\pi_{\tau}}}$ be a sequence of independent random variables that are uniformly distributed on the torus. We choose
\begin{align*}
	&\xi^{\tau,\kappa}=\sqrt{d_{\pi_{\tau}}}\,  \epsilon_{\tau,\kappa}\cdot e_{\kappa}\in \C^{d_{\pi_{\tau}}}\,\,\,\,\,\,\,\,\forall 1\leq \tau\leq T,\, 1\leq \kappa < m_{\pi}(\pi_{\tau}),\\
	& \xi^{\tau,\kappa}=\sqrt{\frac{d_{\pi_{\tau}}}{d_{\pi_{\tau}}-\kappa+1}}\,\sum_{\iota=\kappa}^{d_{\pi_{\tau}}} \epsilon_{\tau,\iota}\cdot e_{\iota}\in \C^{d_{\pi_{\tau}}}\,\,\,\,\,\,\,\,\forall 1\leq \tau\leq T,\,\kappa = m_{\pi}(\pi_{\tau}).
\end{align*}
This gives the measurement matrix
\begin{align}\label{eq:phi_for_mainresult2}
\Phi=\frac{1}{\sqrt{m}} R_{\Omega} \big(\pi(g)\xi\big)_{g\in G}^{\ast} B
\end{align}
where $B\in\C^{n\times n}$ is unitary and $\Omega=(\omega_1,\ldots,\omega_m)$ is a sequence of independent random variables with $\omega_i \sim \mathcal{U}(G)$, i.e. according to the normalized counting measure on $G$, such that $\Omega$ and $\xi$ are independent.\par
The main idea in order to establish our recovery result is to first prove that with high probability the columns of $\big(\pi(g)\xi\big)^{\ast}_{g\in G} B$ form a BOS and then to use a known recovery result for BOS. A similar approach was employed in \cite{romberg2009compressive}.

\begin{proposition}\label{prop:columns_ONB}
	The columns of $(\pi(g)\xi)_{g\in G}^{\ast}B$ are an orthonormal system with respect to the normalized counting measure.
	\begin{proof}
		With regard to Example \ref{ex:discrete_OS} it is enough to show that the rows
		\begin{equation*}
			\left\{\frac{1}{\sqrt{\vert G\vert}} \Big((\pi(g)\xi)_{g\in G}\Big)_{1,-},\ldots, \frac{1}{\sqrt{\vert G\vert}} \Big((\pi(g)\xi)_{g\in G}\Big)_{n,-}\right\}
		\end{equation*}
		have norm 1 and are orthogonal. Therefore, consider $j,k\in\{1,\ldots,n\}$ and calculate
		\begin{align*}
			\left\langle \left(\frac{1}{\sqrt{\vert G\vert}} \Big((\pi(g)\xi)_{g\in G}\Big)_{k,-}\right)^T, \left(\frac{1}{\sqrt{\vert G\vert}} \Big((\pi(g)\xi)_{g\in G}\Big)_{j,-}\right)^T\right\rangle&=\frac{1}{\vert G\vert}\sum_{g\in G} \left(\pi(g)\xi\right)_k\overline{\left(\pi(g)\xi\right)_j}\\
			&=\frac{1}{\vert G\vert}\sum_{g\in G}\sum_{l_1,l_2=1}^n \pi(g)_{kl_1}\xi_{l_1} \overline{\pi(g)_{jl_2}\xi_{l_2}}\\
			&=\frac{1}{\vert G\vert}\sum_{l_1,l_2=1}^n \xi_{l_1} \overline{\xi_{l_2}} \,\sum_{g\in G} \pi(g)_{kl_1} \overline{\pi(g)_{jl_2}}.
		\end{align*}
		We define $\tau(j)=\alpha_{\pi}^{-1}(j)_1$, $\kappa(j)=\alpha_{\pi}^{-1}(j)_2$ and $\iota(j)=\alpha_{\pi}^{-1}(j)_3$, and similarly for $k$. Then, Schur's orthogonality relations give
		\begin{align*}
			\frac{1}{\vert G\vert}\sum_{l_1,l_2=1}^n \xi_{l_1} \overline{\xi_{l_2}} \,\sum_{g\in G} \pi(g)_{kl_1} \overline{\pi(g)_{jl_2}} &= \frac{1}{\vert G\vert}\sum_{\iota_1=1}^{d_{\pi_{\tau(k)}}}\sum_{\iota_2=1}^{d_{\pi_{\tau(j)}}} \xi_{\iota_1}^{\tau(k),\kappa(k)} \overline{\xi_{\iota_2}^{\tau(j),\kappa(j)}} \sum_{g\in G} \pi_{\tau(k),\kappa(k)}(g)_{\iota(k),\iota_1} \overline{\pi_{\tau(j),\kappa(j)}(g)_{\iota(j),\iota_2}}\\
			&=\frac{1}{\vert G\vert}\sum_{\iota_1=1}^{d_{\pi_{\tau(k)}}}\sum_{\iota_2=1}^{d_{\pi_{\tau(j)}}} \xi_{\iota_1}^{\tau(k),\kappa(k)} \overline{\xi_{\iota_2}^{\tau(j),\kappa(j)}}\, \delta_{\tau(j),\tau(k)}\frac{\vert G\vert}{d_{\pi_{\tau(j)}}}\delta_{\iota(k), \iota(j)} \delta_{\iota_1, \iota_2}\\
			&=\frac{1}{d_{\pi_{\tau(j)}}}\sum_{\iota=1}^{d_{\pi_{\tau(j)}}} \xi_{\iota}^{\tau(j),\kappa(k)} \overline{\xi_{\iota}^{\tau(j),\kappa(j)}}\, \delta_{\tau(j),\tau(k)}\delta_{\iota(k), \iota(j)} 
		\end{align*}
		We will now distinguish between three cases. First, assume that $j$ and $k$ correspond to the same block within the block diagonal form of $\pi$. This is equivalent to $\tau(j)=\tau(k)$ and $\kappa(j)=\kappa(k)$ where the function $\alpha_{\pi}$ was defined in the proof of Proposition \ref{prop:C_for_subrep_of_left_reg}. Since we constructed $\xi$ such that $\Vert\xi^{\tau(j),\kappa(j)}\Vert_2^2=d_{\pi_{\tau(j)}}$ and we have that $\delta_{\iota(k), \iota(j)}=\delta_{j,k}$ for $j$ and $k$ corresponding to the same block, we get
		\begin{align*}
			\frac{1}{d_{\pi_{\tau(j)}}}\sum_{\iota=1}^{d_{\pi_{\tau(j)}}} \xi_{\iota}^{\tau(j),\kappa(k)} \overline{\xi_{\iota}^{\tau(j),\kappa(j)}}\, \delta_{\tau(j),\tau(k)}\delta_{\iota(k), \iota(j)}=\left\Vert\xi^{\tau(j),\kappa(j)}\right\Vert_2^2 \,\cdot\frac{1}{d_{\pi_{\tau(j)}}}\delta_{\iota(k), \iota(j)}=\delta_{j,k}.
		\end{align*}
		Now consider the case where $j$ and $k$ correspond to different blocks of $\pi$ and where the two blocks don't consist of the same representation. That is equivalent to $\tau(j)\neq \tau(k)$. Hence,
		\begin{align*}
			\frac{1}{d_{\pi_{\tau(j)}}}\sum_{\iota=1}^{d_{\pi_{\tau(j)}}} \xi_{\iota}^{\tau(j),\kappa(k)} \overline{\xi_{\iota}^{\tau(j),\kappa(j)}}\, \delta_{\tau(j),\tau(k)}\delta_{\iota(k), \iota(j)}=0.
		\end{align*}
		It remains to consider the case where $j$ and $k$ correspond to different blocks of $\pi$ but these blocks  consist of the same representation. This means that $\tau(j)=\tau(k)$ and $\kappa(j)\neq \kappa(k)$. The construction of $\xi$ gives $\langle \xi^{\tau(j),\kappa(k)}, \xi^{\tau(j),\kappa(j)}\rangle =0$ and hence,
		\begin{align*}
			\frac{1}{d_{\pi_{\tau(j)}}}\sum_{\iota=1}^{d_{\pi_{\tau(j)}}} \xi_{\iota}^{\tau(j),\kappa(k)} \overline{\xi_{\iota}^{\tau(j),\kappa(j)}}\, \delta_{\tau(j),\tau(k)}\delta_{\iota(k), \iota(j)}
			&=\frac{1}{d_{\pi_{\tau(j)}}}\langle \xi^{\tau(j),\kappa(k)}, \xi^{\tau(j),\kappa(j)}\rangle  \delta_{\iota(k), \iota(j)}=0,
		\end{align*}	
		where we used the construction of $\xi$. This concludes the proof.
	\end{proof}
\end{proposition}

The next step is to show that the set of columns is a bounded system with high probability.

\begin{proposition}\label{prop:BOS_condition_with_high_prob}
	Let $\delta\in(0,1)$. Then with probability at least $1-\delta$, it holds
	\begin{align*}
		\max_{1\leq j\leq n,\,h\in G} \left\vert \langle \big((\pi(g)\xi)_{g\in G}\big)e_h, Be_j \rangle\right\vert < \sqrt{2d_{\max}(\pi)\ln\left(\frac{2n\vert G\vert}{\delta}\right)}
	\end{align*}
	with
	\begin{align*}
	d_{\max}(\pi):=\begin{cases}
		\max_{\rho\leq \pi:\,m_{\pi}(\rho)>1} d_{\rho},& \exists \rho\leq \pi:\,m_{\pi}(\rho)>1,\\
		1,& \text{otherwise.}
	\end{cases}
	\end{align*}
	\begin{proof}
		First, we notice that the mapping $\beta_{\pi}\colon\{(\tau,\iota)\mid \tau\in\{1,\ldots,T\},\, \iota\in \{1,\ldots, d_{{\pi}_{\tau}}\}\}\to\text{supp}( \xi)$ defined by
		\begin{align*}
		(\tau,\iota)\mapsto \sum_{t=1}^{\tau-1} d_{\pi_t}m_{\pi}(\pi_t)+\min\{m_{\pi}(\pi_{\tau})-1,\iota-1\}d_{\pi_{\tau}}+\iota
		\end{align*}
		is well defined and bijective.
		Thus, we can rewrite the inner product as sum of independent random variables
		\begin{align*}
			\langle\big(\pi(g)\xi)_{g\in G}\big)e_h, Be_j \rangle&=\langle\pi(h)\xi, Be_j \rangle=\sum_{l=1}^n \sum_{k=1}^n \pi(h)_{kl}\,\xi_l \,\overline{(Be_j)_k}=\sum_{l=1}^n \overline{(\pi(h)^{\ast}Be_j)_l} \,\xi_l\\
			&=\sum_{\tau=1}^T \sum_{\iota=1}^{d_{\pi_{\tau}}} \overline{(\pi(h)^{\ast}Be_j)_{\beta_{\pi}(\tau,\iota)}} \,\xi_{\beta_{\pi}(\tau,\iota)}=\sum_{\tau=1}^T \sum_{\iota=1}^{d_{\pi_{\tau}}} A_{\tau,\iota}\epsilon_{\tau,\iota}.
		\end{align*}
		where we used the coefficient vector $A\in\C^{\sum_{\tau=1}^{T}d_{\pi_{\tau}}}$ which is defined by
		\begin{align*}
		A_{\tau,\iota}:=\begin{cases}
		\sqrt{d_{\pi_{\tau}}}\cdot\overline{(\pi(h)^{\ast}Be_j)_{\beta_{\pi}(\tau,\iota)}}, &\iota\in\{1,\ldots,m_{\pi}(\pi_{\tau})-1\},\\
		 \sqrt{\frac{d_{\pi_{\tau}}}{d_{\pi_{\tau}}-m_{\pi}(\pi_{\tau})+1}}\cdot \overline{(\pi(h)^{\ast}Be_j)_{\beta_{\pi}(\tau,\iota)}}, &\iota\in\{m_{\pi}(\pi_{\tau}),\ldots,d_{\pi_{\tau}}\}.\\
		\end{cases}
		\end{align*}
		In order to use a complex version of Hoeffding's inequality we have the $2$-norm of $A$. So,
		\begin{align*}
			\left\Vert \left(A_{\tau,\iota}\right)_{\tau=1,\ldots,T,\,\iota=1,\ldots,d_{\pi_{\tau}}}\right\Vert_2^2\leq d_{\max}(\pi)\Vert (\pi(h)^{\ast}Be_j\Vert_2^2=d_{\max}(\pi)\Vert Be_j\Vert_2^2=d_{\max}(\pi)
		\end{align*}
		where we used that $\pi(h)$ and $B$ are unitary as well as the definition of $d_{\max}(\pi)$. The complex version of Hoeffding's inequality \cite[Corollary 8.10]{foucart2013mathematical} yields
		\begin{align*}
			\mathbb{P}\left( \bigg\vert \langle\big(\pi(g)\xi)_{g\in G}\big)e_h, Be_j \rangle \bigg\vert \geq u\right)\leq 2e^{-\frac{1}{2}\frac{u^2}{d_{\max}(\pi)}}.
		\end{align*}
		for all $u>0$. Taking the union bound over all choices of $h\in G$ and $j\in\{1,\ldots,n\}$ gives
		\begin{align*}
			\mathbb{P}\left( \max_{1\leq j\leq n,\,h\in G}  \bigg\vert \langle\big(\pi(g)\xi)_{g\in G}\big)e_h, Be_j \rangle \bigg\vert \geq u\right) \leq 2n \vert G\vert e^{-\frac{1}{2}\frac{u^2}{d_{\max}(\pi)}}.
		\end{align*}
		Choosing $u=\sqrt{2 d_{\max}(\pi)\ln\left(\frac{2n\vert G\vert}{\delta}\right)}$ finishes the proof.
	\end{proof}
\end{proposition}

\subsection{RIP for measurement matrices with randomized sampling set and generating vector}\label{sec:second_main_result}
	
	We are now ready to state and prove the main result of section 3. It establishes the restricted isometry property for $\Phi$ as it was defined in (\ref{eq:phi_for_mainresult2}). The primary advantage of Theorem \ref{mainresult2} over Theorem \ref{mainresult1} is that the basis $B$, in which the vector $x$ is sparse, can be arbitrary. This indirectly addresses the issue we had in Section \ref{sect:fix_zuf}, where different realizations of a representation exhibited varying recovery properties.
	
	\begin{theorem}\label{mainresult2}
	Let the measurement matrix
	\[
	\Phi=\frac{1}{\sqrt{m}} R_{\Omega}\left(\pi(g)\xi\right)_{g\in G}^{\ast} B\in\C^{m \times n}
	\]
	be as in Section \ref{subsec:likely_BOS}. If, for $\eta,\delta_1,\delta_2\in(0,1)$,
	\begin{align}
		\frac{m}{\ln(9m)}&\geq C_1\delta_1^{-2} d_{\max}(\pi) s \ln\left(\frac{2 n \vert 	G\vert}{1-\sqrt{1-\eta}}\right) \ln(4s)^2\ln(8n),\label{eq:number_measurements_mainresult2_1}\\
		m&\geq C_2 \delta_2^{-2} d_{\max}(\pi)s \ln\left(\frac{2n \vert G\vert}{1-\sqrt{1-\eta}}\right)  \ln\left(\frac{1}{1-\sqrt{1-\eta}}\right)\label{eq:number_measurements_mainresult2_2},
	\end{align}
	then with probability at least $1-\eta $ the restricted isometry constant $\delta_s$ of $\Phi$ satisfies \[\delta_s\leq \delta_1+\delta_1^2+\delta_2.
	\]
	Here, $C_1,C_2>0$ are absolute constants.
	\begin{proof}
		First, let's establish some notation. Whenever we have a random variable $Z=f(X,Y)$ which is defined by two independent random variables $X$ and $Y$ and some measurable function $f$, we write $Z^{X=x_0}$ for the random variable $f(x_0,Y)$. We will use \cite[Theorem 12.32]{foucart2013mathematical}. With a suitable choice of $C_1$ and $C_2$ it follows that for every realization $\xi_0$ of the random variable $\xi$ with
		\begin{align}\label{eq:BOS_condition_for_fixed_xi_uniform}
			\max_{1\leq j\leq n,\,h\in G} \bigg\vert \langle \big(\left(\pi(g)\xi_0\right)_{g\in G} \big)e_h,B e_j\rangle\bigg\vert < \sqrt{2\,d_{\max}(\pi)}\sqrt{\ln\left(\frac{2n\vert G\vert}{1-\sqrt{1-\eta}}\right)}
		\end{align}
		\cite[Theorem 12.32]{foucart2013mathematical} yields
		\begin{align}\label{eq:theo12.32_statement}
			\mathbb{P}\left( \delta_s^{\xi=\xi_0}\leq \delta_1+\delta_1^2+\delta_2 \right)\geq 1-\left(1-\sqrt{1-\eta}\right)=\sqrt{1-\eta}.
		\end{align}
		Here, we have used the fact that the columns of $(\pi(g)\xi)^{\ast}_{g\in G} B$ are an orthonormal system as proven in Proposition \ref{prop:columns_ONB} as well as the condition on the number of measurements in (\ref{eq:number_measurements_mainresult2_1}) and (\ref{eq:number_measurements_mainresult2_2}).
		Now, the idea is to integrate over all possible realizations $\xi_0$ of $\xi$ that satisfy (\ref{eq:BOS_condition_for_fixed_xi_uniform}). Let $f_{\Omega}$ be the probability density function of $\Omega$ with respect to the measure $\nu=\left(\bigotimes_{j=1}^m \mu\right)$ where $\mu$ is the normalized counting measure on $G$, i.e.
		\begin{align*}
			\frac{1}{\vert G\vert^m}=\mathbb{P}(\Omega=\Omega_0)=\intpar{\{\Omega_0\}}{}{f_{\Omega}}{\nu} 
		\end{align*}
		for every $\Omega_0\in G^m$.
		Further, let $f_{\xi}$ be the probability density functions of $\xi$. Thus,
		\begin{align*}
		&\intpar{\{\xi_0\in\C^n \,\mid\, \xi_0 \text{ satisfies } (\ref{eq:BOS_condition_for_fixed_xi_uniform}) \}}{}{\mathbb{P}\Big( \delta_s^{\xi=y} \leq \delta_1+\delta_1^2+\delta_2\Big) f_{\xi}(y)}{y}\\
			&=\intpar{\{\xi_0\in\C^n \,\mid\, \xi_0 \text{ satisfies } (\ref{eq:BOS_condition_for_fixed_xi_uniform}) \}}{}{\intpar{\big\{\Omega_0\in G^m\,\big\vert\, \delta_s^{\Omega=\Omega_0,\,\xi=y} \leq \delta_1+\delta_1^2+\delta_2\big\}}{}{ f_{\Omega}(t) f_{\xi}(y)}{\nu(t)}}{y}\\
			&=\intpar{\{\xi_0\in\C^n \,\mid\, \xi_0 \text{ satisfies } (\ref{eq:BOS_condition_for_fixed_xi_uniform}) \}}{}{\intpar{\big\{\Omega_0\in G^m\,\big\vert\, \delta_s^{\Omega=\Omega_0,\,\xi=y} \leq \delta_1+\delta_1^2+\delta_2\big\}}{}{ f_{\Omega,\xi}(t,y) }{\nu(t)}}{y}\\
			&=\intpar{\big\{(\Omega_0,\xi_0)\in G^m \times \C^n\,\big\vert\, \xi_0 \text{ satisfies } (\ref{eq:BOS_condition_for_fixed_xi_uniform}) \text{ and } \delta_s^{\Omega=\Omega_0,\,\xi=\xi_0} \leq \delta_1+\delta_1^2+\delta_2\big\}}{}{ f_{\Omega,\xi}(t,y)}{(\nu \otimes \leb^{n})(t,y)}
		\end{align*}
		where we used the independence of the random variables $\Omega$ and $\xi$ as well as Fubini's theorem. The following inequality is immediate
		\begin{align*}
			&\intpar{\big\{(\Omega_0,\xi_0)\in G^m \times \C^n\,\big\vert\, \xi_0 \text{ satisfies } (\ref{eq:BOS_condition_for_fixed_xi_uniform}) \text{ and } \delta_s^{\Omega=\Omega_0,\,\xi=\xi_0} \leq \delta_1+\delta_1^2+\delta_2\big\}}{}{ f_{\Omega,\xi}(t,y)}{(\nu \otimes \leb^{n})(t,y)}\\
			&\leq\intpar{\big\{(\Omega_0,\xi_0)\in G^m \times \C^n\,\big\vert\, \delta_s^{\Omega=\Omega_0,\,\xi=\xi_0} \leq \delta_1+\delta_1^2+\delta_2\big\}}{}{ f_{\Omega,\xi}(t,y)}{(\nu \otimes \leb^{n})(t,y)}.
		\end{align*}
		Now using the above as well as (\ref{eq:theo12.32_statement}) and Proposition \ref{prop:BOS_condition_with_high_prob} finishes the proof since we have
		\begin{align*}
			&\mathbb{P}\left(\delta_s\leq \delta_1+\delta_1^2+\delta_2 \right)\\
			&\geq \intpar{\{\xi_0\in\C^n \,\mid\, \xi_0 \text{ satisfies } (\ref{eq:BOS_condition_for_fixed_xi_uniform}) \}}{}{\mathbb{P}\Big( \delta_s^{\xi=y} \leq \delta_1+\delta_1^2+\delta_2\Big) f_{\xi}(y)}{y}\\
			&\geq \intpar{\{\xi_0\in\C^n \,\mid\, \xi_0 \text{ satisfies } (\ref{eq:BOS_condition_for_fixed_xi_uniform}) \}}{}{\sqrt{1-\eta}\, f_{\xi}(y)}{y}\\
			&=\sqrt{1-\eta}\,\,\mathbb{P}\left(\max_{1\leq j\leq n,\,h\in G} \bigg\vert \langle \big(\left(\pi(g)\xi\right)_{g\in G} \big)e_h,B e_j\rangle\bigg\vert < \sqrt{2\,d_{\max}(\pi)}\sqrt{\ln\left(\frac{2n\vert G\vert}{1-\sqrt{1-\eta}}\right)}\right)\\
			&\geq 1-\eta.
		\end{align*}
	\end{proof}
	\end{theorem}
	
	It is important to note that the above Theorem provides especially good bounds on the number of measurements for representations that have either no multiplicities (since then $d_{\max}=1$) or have only small-dimensional subrepresentations. We want to give an example for the latter case: Consider the dihedral group $D_n$. Then, all irreducible representations have dimension 1 or 2 \cite[Section 5.3]{serre1977linear}. Thus, for any $\pi\leq L$ it holds $d_{\max}\leq 2$.\par
	We finish this section, with a reformulation of Theorem \ref{mainresult2}. Since usually one likes to have only one condition on the number of measurements, we slightly weakening the statement of Theorem \ref{mainresult2} to obtain the following corollary. A similar type of reformulation can be found in \cite{foucart2013mathematical}. However, we provide a proof for self-containment in Appendix \ref{subsec:appendix_proof_of_corollary}.
	
	\begin{corollary}\label{cor:one_condition_on_m_result}
		Let the measurement matrix
		\[
		\Phi=\frac{1}{\sqrt{m}} R_{\Omega}\left(\pi(g)\xi\right)_{g\in G}^{\ast} B\in\C^{n \times n}
		\]
		be as in Section \ref{subsec:likely_BOS}. If, for $\eta,\delta \in(0,1)$,
		\begin{align}\label{eq:one_inequ_for_m_uniform_recovery}
			m\geq C \delta^{-2} s \,d_{\max}(\pi) \ln(8\vert G\vert)\ln\left(\frac{2}{\eta}\right) \max\bigg\{ &\ln(4s)^2\ln(8n)\nonumber\\
			&\cdot\ln\left(\delta^{-2}s\,d_{\max}(\pi)\ln(8\vert G\vert)\ln\left(\frac{2}{\eta}\right) \right),\ln\left(\frac{2}{\eta}\right)\bigg\},
		\end{align}
		then with probability at least $1-\eta $ the restricted isometry constant $\delta_s$ of $\Phi$ satisfies \[\delta_s\leq 3\delta.
		\]
		Here, $C>0$ is an absolute constant.
	\end{corollary}

	We note that Theorem \ref{mainresult2_low_key} is an immediate consequence of the above corollary.\par	
	Let's revisit our thoughts from the introduction of Section \ref{sect:random_xi_and_omega}. There, we noted that unitarily equivalent representations can have different $s$-sparse recovery properties, if a sampling set is fixed beforehand. By randomizing $\Omega$ in Theorem \ref{mainresult2}, we have proven a recovery result where the unitary matrix $B$, describing a basis in which $x$ has to be sparse, can be arbitrary. Thus, the question arises whether this fixes our problem. It indeed will.
	\begin{remark}\label{rem:equiv_rep}
	Let $\rho\leq L$ be a unitary representation. Then, there exists a unitary matrix $V\in\C^{n\times n}$ such that $\rho(g)=V^{\ast} \pi(g)V$ for all $g\in G$ with $\pi$ having block-diagonal form as given in (\ref{eq:block_diag_form_pi}). Fix a unitary matrix $B\in\C^{n\times n}$. Theorem \ref{mainresult2} implies that
	\[
	\Phi_{\pi}=\frac{1}{\sqrt{m}} R_{\Omega} \big(\pi(g)\xi\big)_{g\in G}^{\ast}B
	\]
	does $s$-sparse recovery with high probability where $\xi$ is as in Section \ref{subsec:likely_BOS}. Equation (\ref{eq:relation_meas_mat_equiv_rep}) implies for $\xi_{\rho}=V^{\ast}\xi$ that
	\begin{align*}
	 	\Phi_{\rho}=\frac{1}{\sqrt{m}} R_{\Omega} \big(\rho(g)\xi_{\rho}\big)_{g\in G}^{\ast} B =\frac{1}{\sqrt{m}} R_{\Omega} \big(\pi(g)\xi\big)_{g\in G}^{\ast}VB.
	\end{align*}
	Again, Theorem \ref{mainresult2} implies that $\Phi_{\rho}$ does $s$-sparse recovery with high probability. Hence, $\pi$ and $\rho$ have the same recovery properties. It follows that unitarily equivalent representations have the same $s$-sparse recovery properties when $\Omega$ is randomized. This also answers the question why most of our assumptions on $\pi$ in Section \ref{subsec:likely_BOS} were not restrictive.
	\end{remark}

\appendix
	\section{}
	\subsection{Induced representations}\label{subsec:appendix_induced_rep}
	We present the necessary background about induced representations.
	\begin{definition}
		Let $H$ be subgroup of $G$ and let $\sigma\colon H \to \text{GL}(W)$ be a representation of $H$. Define a vector space
		\begin{align*}
			V=\{f\colon G\to W\,\vert\, f(hg)=\sigma(h)f(g)\,\,\,\,\forall h\in H, g\in G\}.
		\end{align*}
		Further, define the group homomorphism $\textnormal{Ind}_{H}^G\sigma\colon G\to \text{GL}(V)$ by
		\begin{align*}
			(\textnormal{Ind}_H^G\sigma (g_1)f)(g_2)=f(g_2g_1)
		\end{align*}
		for all $g_1,g_2\in G$. $\textnormal{Ind}_{H}^G\sigma$ is called the induced representation from $H$ up to $G$. 
	\end{definition}

	It is easy to check that $\textnormal{Ind}_H^G$ is indeed a representation.	Now, our goal is to give a specific  equivalent representation to $\textnormal{Ind}_H^G\sigma$. The construction follows the arguments presented in \cite[Chapter 2]{kaniuth2013induced}.\par
	Fix a cross-section $\gamma\colon H\setminus G\to G$, i.e. a mapping such that $\gamma(\upsilon)\in \upsilon$ for every $\upsilon\in H\setminus G$, and define $T\colon W^{H\setminus G}\to V$ by
	\begin{align*}
		(Tf)(h\gamma(\upsilon))=\sigma(h)f(\upsilon).
	\end{align*}
	This is well-defined: For $g\in G$ there is a unique $\upsilon\in H\setminus G$ such that $g\in \upsilon$. So there is a unique $h\in H$ such that $g=h\gamma(\upsilon)$. So, $g$ can be uniquely written as $h\gamma(\upsilon)$.	It remains to check that $Tf\in V$ for every $f\in W^{H\setminus G}$. Let $f\in W^{H\setminus G}$, $h\in H$ and $g=\tilde{h}\gamma(\upsilon)\in G$ as before. Then,
	\begin{align*}
		(Tf)(hg)=(Tf)\!\left(h\tilde{h}\gamma(\upsilon)\!\right)=\sigma\!\left(h\tilde{h} \right)\!f(\upsilon)=\sigma(h)\sigma\!\left(\tilde{h} \right)\!f(\upsilon)=\sigma(h)(Tf)\!\left(\tilde{h}\gamma(\upsilon)\!\right)=\sigma(h)(Tf)\left(g\right).
	\end{align*}
	It is easy to check that $T$ is linear. Now, let $f_1,f_2\in W^{H\setminus G}$ with $Tf_1=Tf_2$. This implies that
	\begin{align*}
		f_1(\upsilon)=\sigma(e)f_1(\upsilon)=(Tf_1)(e\gamma(\upsilon))=(Tf_2)(e\gamma(\upsilon))=\sigma(e)f_2(\upsilon)=f_2(\upsilon)
	\end{align*}
	for all $\upsilon\in H\setminus G$. Thus, $f_1=f_2$ and $T$ is injective. Lastly, we show that $T$ is surjective. Let $F\in V$. Then, define $f\colon H\setminus G\to W,\, \upsilon\mapsto F(\gamma(\upsilon))$. It holds
	\begin{align*}
		(Tf)(h\gamma(\upsilon))=\sigma(h)f(\upsilon)=\sigma(h)F(\gamma(\upsilon))=F(h\gamma(\upsilon))
	\end{align*}
	for all $h\in H$ and $\upsilon\in H\setminus G$.\par
	Since $T$ is bijective, it has an inverse $T^{-1}\colon V\to W^{H\setminus G}$ and it is easy to check that $T^{-1}$ is given by $T^{-1}F=F\circ \gamma$. Now, we want to determine $T^{-1}(\textnormal{Ind}_H^G\sigma)(g) T$ for every $g\in G$. Let $g\in G$, $f\in W^{H\setminus G}$ and $\upsilon \in H\setminus G$. First note that
	\begin{align*}
		\gamma(\upsilon)g=\left( \gamma(\upsilon)g \,\gamma(\upsilon g)^{-1} \right)\gamma(\upsilon g)
	\end{align*}
	with $\gamma(\upsilon)g \,\gamma(\upsilon g)^{-1}\in H$. Thus, we get
	\begin{align*}
		\big(T^{-1}(\textnormal{Ind}_H^G\sigma)(g) Tf\big)(\upsilon)&=\big((\textnormal{Ind}_H^G\sigma)(g) Tf\big)(\gamma(\upsilon))\\
		&=(Tf)(\gamma(\upsilon)g)\\
		&=(Tf)\left(\left( \gamma(\upsilon)g \,\gamma(\upsilon g)^{-1} \right)\gamma(\upsilon g)\right)\\
		&=\sigma\left( \gamma(\upsilon)g \,\gamma(\upsilon g)^{-1} \right)f(\upsilon g).
	\end{align*}
	Now assume that $W$ is an inner product space and $\sigma$ is unitary. Then, $V$ becomes an inner product space via
	\begin{align*}
		\langle f_1,f_2\rangle_V:=\sum_{Hg\in H\setminus G}\langle f_1(g),f_2(g)\rangle_W.
	\end{align*}
	Let us first check that the above sum is well-defined. Let $f_1,f_2\in V$ and let $\gamma_1,\gamma_2\colon H\setminus G\to G$ be two cross-sections. For every $\upsilon\in H\setminus G$ there exists a unique $h_{\upsilon}\in H$ such that $h_{\upsilon}\gamma_2(\upsilon)=\gamma_1(\upsilon)$. Then,
	\begin{align*}
		\sum_{\upsilon\in H\setminus G} \langle f_1(\gamma_1(\upsilon)), f_2(\gamma_1(\upsilon))\rangle_W &=\sum_{\upsilon\in H\setminus G} \langle f_1(h_{\upsilon}\gamma_2(\upsilon)), f_2(h_{\upsilon}\gamma_2(\upsilon))\rangle_W\\
		&=\sum_{\upsilon\in H\setminus G} \langle \sigma(h_{\upsilon}) f_1(\gamma_2(\upsilon)), \sigma(h_{\upsilon}) f_2(\gamma_2(\upsilon))\rangle_W\\
		&=\sum_{\upsilon\in H\setminus G} \langle  f_1(\gamma_2(\upsilon)),  f_2(\gamma_2(\upsilon))\rangle_W
	\end{align*}
	since $f_1,f_2\in V$ and $\sigma$ is unitary. Thus, $\langle\cdot,\cdot\rangle_V$ is well-defined. The inner product properties are obvious.	The representation $\textnormal{Ind}_H^G \sigma$ is unitary with respect to this inner product since
	\begin{align*}
			\langle \textnormal{Ind}_H^G\sigma(g)f_1, \textnormal{Ind}_H^G\sigma(g)f_2\rangle_V&=\sum_{\upsilon\in H\setminus G} \langle \textnormal{Ind}_H^G\sigma(g) f_1(\gamma(\upsilon)),\textnormal{Ind}_H^G\sigma(g) f_2(\gamma(\upsilon))\rangle_W\\
			&=\sum_{\upsilon\in H\setminus G} \langle f_1(\gamma(\upsilon )g)),f_2(\gamma(\upsilon )g))\rangle_W\\
			&=\sum_{\upsilon\in H\setminus G} \langle f_1(\gamma(\upsilon ))),f_2(\gamma(\upsilon )))\rangle_W\\
			&=\langle f_1,f_2\rangle_V
	\end{align*}
	for all $f_1,f_2\in V$ where we used that $H\setminus G\to H\setminus G,\,\upsilon\to\gamma(\upsilon)g$ is bijective.	Further, T becomes an isometry: For all $f_1,f_2\in W^{H\setminus G}$ we get
	\begin{align*}
		\langle Tf_1,Tf_2\rangle_V=\sum_{\upsilon\in H\setminus G} \langle Tf_1(\gamma(\upsilon)),Tf_2(\gamma(\upsilon))\rangle_W=\sum_{\upsilon\in H\setminus G} \langle f_1(\upsilon)),f_2(\upsilon))\rangle_W=:\langle f_1,f_2\rangle_{W^{H\setminus G}},
	\end{align*}
	since the scalar product $\langle\cdot,\cdot\rangle_V$ is independent of the chosen summation.
	
	\subsection{Proof of Corollary \ref{cor:one_condition_on_m_result}}\label{subsec:appendix_proof_of_corollary}
		The following inequality will be helpful to prove the statement
		\begin{align}\label{eq:helpful_ineq_1}
			\ln\left( \!\frac{2n\vert G\vert}{1-\sqrt{1-\eta}} \!\right)\!\leq\! 2\ln(2\vert G\vert)+\ln\left(\!\frac{1}{1-\sqrt{1-\eta}}\!\right)\!\leq\! 2\ln(2\vert G\vert)+ \ln\left(\!\frac{2}{\eta}\!\right)\!\leq\! 2\ln(8\vert G\vert) \ln\left(\!\frac{2}{\eta}\!\right).
		\end{align}
		We start by proving that (\ref{eq:one_inequ_for_m_uniform_recovery}) already implies (\ref{eq:number_measurements_mainresult2_1}) for some suitable universal constant $C>0$.  The function $x\mapsto \frac{x}{\ln(9x)}$ is increasing on $[1,\infty)$. Thus, the assumption on the number of measurements implies
		\begin{align*}
			\frac{m}{\ln(9m)}&\geq \frac{C \delta^{-2} s d_{\max}(\pi) \ln(8\vert G\vert)\ln\left(\frac{2}{\eta}\right)  \ln(4s)^2\ln(8n) \ln\left(\delta^{-2}sd_{\max}(\pi)\ln(8\vert G\vert)\ln\left(\frac{2}{\eta}\right) \right)}{\ln\left( 9C \delta^{-2} s d_{\max}(\pi) \ln(8\vert G\vert)\ln\left(\frac{2}{\eta}\right)  \ln(4s)^2\ln(8n) \ln\left(\delta^{-2} s d_{\max}(\pi)\ln(8\vert G\vert)\ln\left(\frac{2}{\eta}\right) \right)\right)}\\
			&=C \delta^{-2} s d_{\max}(\pi) \ln(8\vert G\vert)\ln\left(\frac{2}{\eta}\right)  \ln(4s)^2\ln(8n)\\
			&\cdot \frac{\ln\left(\delta^{-2}s d_{\max}(\pi)\ln(8\vert G\vert)\ln\left(\frac{2}{\eta}\right) \right)}{\ln\left( 9C \delta^{-2} s d_{\max}(\pi) \ln(8\vert G\vert)\ln\left(\frac{2}{\eta}\right)  \ln(4s)^2\ln(8n) \ln\left(\delta^{-2}s d_{\max}(\pi)\ln(8\vert G\vert)\ln\left(\frac{2}{\eta}\right) \right)\right)}
		\end{align*}
		Now, using the following two inequalities
		\begin{align*}
			&\ln\left(9C\right)\leq \ln\left(9C\right)3 \ln\left(\delta^{-2}s d_{\max}(\pi)\ln(8\vert G\vert)\ln\left(\frac{2}{\eta}\right) \right) \quad\text{and}\\
			& \ln\left(\ln(4s)^2\right)=2\ln(\ln(4s))\stackrel{s\geq 1}{\leq} 2\ln(s\ln(8)\ln(2))\leq 2\ln\left(\delta^{-2}s d_{\max}(\pi)\ln(8\vert G\vert)\ln\left(\frac{2}{\eta}\right) \right)
		\end{align*}
		as well as the inequality (\ref{eq:helpful_ineq_1}) gives
		\begin{align*}
			\ln\left(\frac{m}{9m}\right)&\geq C \delta^{-2} s d_{\max}(\pi) \ln(8\vert G\vert)\ln\left(\frac{2}{\eta}\right)  \ln(4s)^2\ln(8n)\frac{1}{3\ln(9C)+1+2+3+1}\\
			&\geq \frac{C}{6\ln(9C)+14} \delta^{-2} s d_{\max}(\pi)  \ln(4s)^2\ln(8n) \ln\left( \frac{2n\vert G\vert}{1-\sqrt{1-\eta}} \right).
		\end{align*}
		So choosing $C>0$ large enough, i.e. such that
		\begin{align*}
			\frac{C}{6\ln\left(9C\right)+14}\geq C_1,
		\end{align*}
		establishes inequality (\ref{eq:number_measurements_mainresult2_1}).\\
		Now consider the second required inequality stated in (\ref{eq:number_measurements_mainresult2_1}). The inequality (\ref{eq:one_inequ_for_m_uniform_recovery}) implies
		\begin{align*}
			m&\geq C \delta^{-2} s d_{\max}(\pi) \ln(8\vert G\vert)\ln\left(\frac{2}{\eta}\right)\ln\left(\frac{2}{\eta}\right)\stackrel{(\ref{eq:helpful_ineq_1})}{\geq} C \delta^{-2} s d_{\max}(\pi)\frac{1}{2}\ln\left( \frac{2n\vert G\vert}{1-\sqrt{1-\eta}} \right) \ln\left(\frac{2}{\eta}\right)\\
			&\geq \frac{C}{2}  \delta^{-2} s d_{\max}(\pi)\ln\left( \frac{2n\vert G\vert}{1-\sqrt{1-\eta}} \right) \ln\left(\frac{1}{1-\sqrt{1-\eta}}\right)
		\end{align*}
		Choosing $C>0$ such that $\frac{C}{2}\geq C_2$ establishes (\ref{eq:number_measurements_mainresult2_1}). The statement follows now from Theorem \ref{mainresult2}.

	\section*{Acknowledgements}
	We thank H. Rauhut (LMU Munich) for fruitful discussion and advice. The authors acknowledge funding by the Deutsche Forschungsgemeinschaft (DFG, German Research Foundation) - Project number 442047500 through the Collaborative Research Center "Sparsity and Singular Structures" (SFB 1481).

	\bibliography{Sparse_Recovery_from_Group_Orbits_paper_arxiv.bib}
	\bibliographystyle{plain}
	\vspace{0.5cm}

\end{document}